\documentclass[10pt,a4paper]{article}

\usepackage{amsmath,amsgen,latexsym}
\usepackage{amstext,amssymb,amsfonts,latexsym}
\usepackage{theorem}
\usepackage{pifont}
\usepackage{graphics,epsfig}
\usepackage{graphicx}
\usepackage{graphicx}

 \newcommand{\bs}{\bigskip}
 \newcommand{\ms}{\medskip}
 \newcommand{\n}{\noindent}
 \newcommand{\s}{\smallskip}
 \newcommand{\hs}[1]{\hspace*{ #1 mm}}
 \newcommand{\vs}[1]{\vspace*{ #1 mm}}

 \newcommand{\setempty}{\emptyset}
 
 \newcommand{\nat}{\mathbb{N}}
 
 \newcommand{\integer}{\mathbb{Z}}

 \newcommand{\dom}{\mbox{dom}}

 \newcommand{\co}{\mathrm{co}\mbox{-}}

 \newcommand{\ie}{\textrm{i.e.},\hspace*{2mm}}
 \newcommand{\eg}{\textrm{e.g.},\hspace*{2mm}}

 \newcommand{\CC}{{\cal C}}
 \newcommand{\FF}{{\cal F}}
 
 \newcommand{\EE}{{\cal E}}

 \newcommand{\GG}{{\cal G}}

 \newcommand{\PP}{{\cal P}}

 \newcommand{\p}{\mathrm{P}}
 \newcommand{\np}{\mathrm{NP}}

 \newcommand{\cfl}{\mathrm{CFL}}

 \newcommand{\comb}[2]{\left(\:\begin{subarray}{c} #1 \\%
      #2 \end{subarray}\right)}

\theoremstyle{plain}
\theoremheaderfont{\bfseries}
 \newtheorem{theorem}{Theorem}[section]
 \newtheorem{lemma}[theorem]{Lemma}
 \newtheorem{proposition}[theorem]{Proposition}

 \newtheorem{openproblem}[theorem]{Question}

 {\theorembodyfont{\rmfamily}
  }
 {\theorembodyfont{\rmfamily} }
 {\theorembodyfont{\rmfamily} }

 \newtheorem{claim}{Claim}

 \newenvironment{proof}{\par \noindent
            {\bf Proof. \hs{2}}}{\hfill$\Box$ \vspace*{3mm}}

 \newenvironment{proofof}[1]{\vspace*{5mm} \par \noindent
         {\bf Proof of #1.\hs{2}}}{\hfill$\Box$ \vspace*{3mm}}

 \newenvironment{yproof}{\par \noindent
            {\bf Proof. \hs{2}}}{\hfill$\Box$ \vspace*{3mm}}

 \newcommand{\floors}[1]{\lfloor #1 \rfloor}

 \newcommand{\sigmap}[1]{\Sigma^{\mathrm{P}}_{#1}}

\newcommand{\ignore}[1]{}

\newcommand{\track}[2]{[\:\begin{subarray}{c} #1 \\%
      #2 \end{subarray} ]}

\newcommand{\cent}{|\!\! \mathrm{c}}
\newcommand{\dollar}{\$}

 \newcommand{\cflsv}{\mathrm{CFLSV}}

 \newcommand{\cflmv}{\mathrm{CFLMV}}

 \newcommand{\cfltwov}{\mathrm{CFL2V}}
 \newcommand{\cflkv}[1]{\mathrm{CFL}{#1}\mathrm{V}}

 \newcommand{\sigmacfl}[1]{\Sigma^{\mathrm{CFL}}_{ #1 }}

 \newcommand{\sigmacflmv}[1]{\Sigma^{\mathrm{CFL}}_{ #1 }\mathrm{MV}}
 \newcommand{\sigmacflsv}[1]{\Sigma^{\mathrm{CFL}}_{ #1 }\mathrm{SV}}

 \newcommand{\picflmv}[1]{\Pi^{\mathrm{CFL}}_{ #1 }\mathrm{MV}}
 \newcommand{\picflsv}[1]{\Pi^{\mathrm{CFL}}_{ #1 }\mathrm{SV}}

 \newcommand{\sigmapmv}[1]{\Sigma^{\mathrm{P}}_{ #1 }\mathrm{MV}}

 \newcommand{\pipmv}[1]{\Pi^{\mathrm{P}}_{ #1 }\mathrm{MV}}

 \newcommand{\ucfltwov}{\mathrm{UCFL2V}}
 \newcommand{\ucflkv}[1]{\mathrm{UCFL}{#1}\mathrm{V}}

\begin{document}
%%%%%%%%%%%%%%%%%%%%%%%%%%%%%%%%%%%%%%%
%%%%%%%%%%%%%%%%%%%%%%%%%%%%%%%%%%%%%%%

\pagestyle{plain}
\setcounter{page}{1}

\begin{center}
{\Large {\bf Not All Multi-Valued Partial CFL Functions Are \s\\
Refined by Single-Valued Functions}}\footnote{An extended abstract appeared in the Proceedings of the 8th IFIP International Conference on Theoretical Computer Science (IFIP TCS 2014), Rome, Italy, September 1--3, 2014, Lecture Notes in Computer Science, Springer, vol. 8705, pp. 136--150.}
\bs\s\\

{\sc Tomoyuki Yamakami}\footnote{Faculty of Engineering,
University of Fukui, 3-9-1 Bunkyo, Fukui 910-8507, Japan} \ms
\end{center}

%%%%%%%%%%%%%%%%%
%%%%%%%%%%%%%%%%%

\begin{quote}
\n{\bf Abstract:}
Multi-valued partial CFL functions are functions computed along accepting computation paths by one-way nondeterministic pushdown automata, equipped with write-only output tapes, which are allowed to reject an input, in comparison with single-valued partial CFL functions. We give an answer to a fundamental question, raised by Konstantinidis, Santean, and Yu [Act. Inform. 43 (2007) 395--417], of whether all such multi-valued partial CFL functions can be refined by  single-valued partial CFL functions.
We negatively solve this open question by presenting a special multi-valued partial CFL function as an example function and by proving that no refinement of this particular function becomes a single-valued partial CFL function. This contrasts an early result of Kobayashi [Inform. Control 15 (1969) 95--109] that multi-valued partial NFA functions are always refined by single-valued NFA functions, where NFA functions are computed by one-way nondeterministic finite automata with output tapes. Our example function turns out to be unambiguously 2-valued, and thus we obtain a stronger separation result, in which no refinement of unambiguously 2-valued partial CFL functions can be single-valued.
For the proof of this fact, we first introduce a new concept of colored automata having no output tapes but having ``colors,''  which can simulate pushdown automata equipped with constant-space output tapes. We then conduct an extensive combinatorial analysis on the behaviors of transition records of stack contents (called stack histories) of these colored automata.

\s

\n{\bf Keywords:}
multi-valued partial function, CFL function, NFA function, refinement, pushdown automaton, context-free language, colored automaton, stack history
\end{quote}

%%%%%%%%%%%%%%%%%%%%%%%%%%%%%%%%%%
%%%%%%%%%%%%%%%%%%%%%%%%%%%%%%%%%%
%%%%%%%%%%%%%%%%%%%%%%%%%%%%%%%%%%
%%%%%%%%%%%%%%%%%%%%%%%%%%%%%%%%%%
\sloppy
%%%%%%%%%%%%%%%%%%%%%%%%%%%%%%%%%%
%%%%%%%%%%%%%%%%%%%%%%%%%%%%%%%%%%
\section{Resolving a Fundamental Question}\label{sec:introduction}

Since early days of automata and formal language theory, multi-valued partial functions,\footnote{Throughout this paper, we often call those multi-valued partial functions just ``functions.''}  computed by various types of automata equipped with supplemental \emph{write-only output tapes}, have been investigated extensively. To keep a restricted nature of memory usage, we require the automata to write output symbols \emph{in an oblivious way}; namely, the automata move their output-tape heads to new blank cells whenever they write non-blank output symbols. We succinctly refer such output tapes to ``write only.''
Among those types of functions, we intend to spotlight
{\em CFL functions} (also known as \emph{algebraic transductions}), which are computed by {\em one-way nondeterministic pushdown automata} (succinctly abbreviated as \emph{npda's}) whose input-tape heads move only in one direction (from the left to the right) with write-only output tapes. Such functions were formally discussed in 1963 by Evey \cite{Eve63} and Fisher \cite{Fis63}. The acronym CFL stands for \emph{context-free languages} because, with no output tapes, the machines recognize precisely context-free languages. Therefore, those functions naturally inherit certain distinctive traits from the context-free languages; however, their behaviors are in essence quite different from the behaviors of the languages. Such intriguing properties of those functions have been addressed occasionally in the past literature (\eg \cite{CC83,Eve63,Fis63,KSY07,Yam11,Yam14b}).

Along their numerous accepting computation paths, npda's can produce various output values on their output tapes. We flexibly allow npda's to reject an input, producing no valid output values.
When the number of output values is always limited to at most one, we obtain {\em single-valued} partial functions. Such single-valued partial functions can be obviously treated
as multi-valued partial functions, but multi-valued partial functions are, in general, not single-valued. For expressing a relationship between multi-valued and single-valued partial functions, it is therefore  more appropriate to ask a question of whether multi-valued partial functions can be {\em refined} by single-valued partial functions, where ``refinement'' is a notion discussed initially for NP functions \cite{Sel94} and it refers to
a certain natural restriction on the outcomes of multi-valued functions.
To be more precise, we say that a function $g$ is a {\em refinement}  (also called ``uniformization'' \cite{KSY07}) of another function $f$  if and only if (i) $f$ and $g$ have the same domain and (ii) for every input $x$ in the domain of $f$, all output values of $g$ on $x$ are also output values of $f$ on the same input $x$.
When $g$ is particularly single-valued, $g$ acts as a ``selection'' function that picks exactly one value out of a set of output values of $f$ on $x$ whenever the set is nonempty. This refinement notion is known to play a significant role also in language recognition.
In a polynomial-time setting, for instance, if we can effectively find an accepting computation path of any polynomial-time nondeterministic Turing machine, then every \emph{multi-valued partial NP function} (which is computed by a certain polynomial-time nondeterministic Turing machine) has a refinement in the form of single-valued NP function. Therefore, this ``no-refinement'' claim for multi-valued partial NP functions immediately
leads to a negative answer to
the long-standing $\p=?\np$ question. More generally, multi-valued partial $\sigmap{k}$-functions in the so-called \emph{NPMV-hierarchy} $\{\sigmapmv{k},\pipmv{k}\mid k\geq1\}$ are not in general refined by  single-valued partial  $\sigmap{k}$-functions  as long as the polynomial(-time) hierarchy forms an infinite hierarchy  \cite{FHOS97,Sel96}.

Returning to automata theory, we can discuss a similar refinement question on CFL functions in hope that we resolve it without any unproven assumption, such as the separation of the polynomial hierarchy. Along this line of research, the first important step was taken by Kobayashi \cite{Kob69} in 1969. He gave an affirmative answer to the refinement question on \emph{multi-valued partial NFA functions}, which are computed by \emph{one-way nondeterministic finite automata} (or \emph{nfa's}, in short) with write-only output tapes; namely, multi-valued partial NFA functions can be refined by appropriate single-valued partial NFA functions.
Konstantinidis, Santean, and Yu \cite{KSY07} discussed a similar question concerning multi-valued partial  CFL functions. They managed to obtain a partial affirmative answer but unfortunately they left the whole question open.

This paper is focused on CFL functions whose output values are produced by npda's that halt in linear time\footnote{This linear time-bound ensures that every CFL function produces only at most an exponential number of output values and it therefore becomes an NP function. This fact naturally extends a well-known containment of $\cfl\subseteq\np$. If no execution time bound is imposed, on the contrary, then a function computed by an npda that nondeterministically produces every string on its output tape on each input also becomes a ``valid'' CFL function but such the function is no longer an NP function.}  (that is, all computation paths terminate in time $O(n)$, where $n$ is the size of input)  with write-only output tapes.
By adopting succinct notations from \cite{Yam11,Yam14b}, we express  as $\cflmv$ the collection of all such CFL functions and we also write $\cflsv$ for a collection of all single-valued partial functions in $\cflmv$.
As a concrete example of our CFL function, let us consider $f$ defined by setting $f(1^n\#x)$ to be the set of all substrings of $x$ of length between $1$ and $n$, exactly when $1\leq n\leq |x|$. This function $f$ is a multi-valued partial CFL function and the following function $g$ is an obvious refinement of $f$; the set $g(1^n\#x)$ is composed only of the first symbol of $x$ whenever $1\leq n\leq|x|$. Notice that $g$ belongs to $\cflsv$.

For a further discussion, it is beneficial to introduce another succinct notation concerning ``refinement.''
Given two classes $\FF$ and $\GG$ of multi-valued partial functions, we write $\FF\sqsubseteq_{ref}\GG$ if every function in $\FF$ can be refined by an appropriately chosen function in $\GG$. Using this notation, the aforementioned refinement question of Konstantinidis et al. regarding CFL functions can be rephrased  neatly as follows.

\begin{openproblem}\label{CFLMV-refine-k=1}
Is it true that $\cflmv\sqsubseteq_{ref}\cflsv$?
\end{openproblem}

Various extensions of $\cflmv$ in Question \ref{CFLMV-refine-k=1} are also possible. We state one such possible extension. Yamakami \cite{Yam14b} lately  introduced a functional hierarchy $\{\sigmacflmv{k},\picflmv{k}\mid k\geq1\}$ (called the \emph{CFLMV hierarchy}), which is  built upon  multi-valued partial CFL functions by applying Turing relativization and a complementation operation (see Section \ref{sec:future}), analogously to the aforementioned NPMV hierarchy $\{\sigmapmv{k},\pipmv{k}\mid k\geq1\}$ over multi-valued partial NP functions \cite{FHOS97,Sel96}. Its single-valued version is customarily denoted by $\{\sigmacflsv{k},\picflsv{k}\mid k\geq1\}$. The function $g$ defined as $g(w)=\{x\in\{0,1\}^*\mid \exists\,u,v\,[w=uxxv]\}$ for each $w\in\{0,1\}^*$ is a simple example of function in $\sigmacflmv{2}$.

Our focal question, Question \ref{CFLMV-refine-k=1}, can be further generalized to the following.

\begin{openproblem}\label{CFLMV-refine-general}
Does $\sigmacflmv{k}\sqsubseteq_{ref}\sigmacflsv{k}$ hold for each index $k\geq1$?
\end{openproblem}

When $k\geq3$, Yamakami \cite{Yam14b} shed partial light on this general question. He was able to show that, for every index $k\geq3$, $\sigmacfl{k-1}=\sigmacfl{k}$ implies  $\sigmacflmv{k}\sqsubseteq_{ref}\sigmacflsv{k}$, where $\sigmacfl{k}$ is the $k$th level of the {\em CFL hierarchy} \cite{Yam14a}, which is the language counterpart of the CFLMV hierarchy. Since the collapse of the CFL hierarchy is closely related to that of the polynomial hierarchy, the answer to Question \ref{CFLMV-refine-general} (when $k\geq3$) might possibly be quite difficult to obtain. See Section \ref{sec:future} for a further discussion.  Nevertheless, the remaining cases of $k=1,2$ have been left unsolved.

In this paper, without relying on any unproven assumption,
we solve Question \ref{CFLMV-refine-general} \emph{negatively} when $k=1$; therefore, our result completely settles Question \ref{CFLMV-refine-k=1}.
Our solution actually gives an essentially stronger statement than what we have discussed so far. To explain this statement, we need a new function class $\cfltwov$ as the collection of all functions $f$ in $\cflmv$ satisfying the condition that the number of output values of $f$ on each input should be at most $2$. We actually obtain the following statement.

\begin{theorem}\label{CFL2V-refine-CFLSV}
$\cfltwov\not\sqsubseteq_{ref} \cflsv$.
\end{theorem}

Since $\cflsv\subseteq \cfltwov\subseteq \cflmv$ holds, Theorem \ref{CFL2V-refine-CFLSV} clearly leads to a negative answer to
Question \ref{CFLMV-refine-k=1}. The proof of this theorem is essentially a manifestation of the following intuition: since an npda relies on limited functionality of its memory device (a stack), along any single computation path, it cannot simulate simultaneously two independent computation paths of another npda.

Instead of providing a detailed proof for Theorem \ref{CFL2V-refine-CFLSV},
we wish to present a simple and clear argument to demonstrate a slightly stronger result regarding a subclass of $\cfltwov$. To justify an introduction of such a subclass, we need to address that even if a function $f$ is single-valued, its underlying npda on each input may have numerous accepting computation paths, each of which produces the same value of $f$.
Hence, controlling the number of those accepting computation paths may be difficult for npda's.
We thus restrict our attention on special npda's that have ``few'' accepting computation paths for each output value.
Let us first call an npda $N$ with a write-only output tape {\em unambiguous} if, for every input $x$ and any output value $y$, $N$ has exactly one accepting computation path producing $y$.
Finally, we denote by $\ucfltwov$ the class of all $2$-valued partial functions computed in linear time by unambiguous npda's equipped with output tapes. Succinctly, those functions are called {\em unambiguously 2-valued}. Obviously, $\ucfltwov\subseteq \cfltwov$ holds.

Throughout this paper, we wish to show the following stronger separation result (than Theorem \ref{CFL2V-refine-CFLSV}), which is referred to as the ``main theorem'' in the subsequent sections.

\begin{theorem}[Main Theorem]\label{UCFL2V-by-CFLSV}
$\ucfltwov\not\sqsubseteq_{ref}\cflsv$.
\end{theorem}

Following a brief explanation of key notions and notation in Section \ref{sec:preliminaries}, we will give in Section \ref{sec:main-theorem} the proof
of Theorem \ref{UCFL2V-by-CFLSV}, completing the proof of Theorem \ref{CFL2V-refine-CFLSV} as well.
Our proof will start  in Sections \ref{sec:example-function} with a presentation of our example function $h_3$, a member of $\ucfltwov$. The proof will then proceed, by way of contradiction, starting with a faulty assumption that a certain refinement, say, $g$ of $h_3$
exists in $\cflsv$.
Thus, there is an npda computing $g$ using a write-only output tape. For our proof, however,  we wish to avoid the messy handling of the output tape of this npda and seek a simpler model of automaton for an easier analysis of its behaviors. For this purpose, we will introduce in Section  \ref{sec:colored-automata} a new concept of ``colored'' automaton---a new type of automaton having no output tape but having ``colors''---which can
simulate any npda equipped with an output tape that computes $g$.
To each accepting computation path of such colored automata, we assign a certain color if the machine pushes the same colored symbols into a stack along this computation path.

To lead to the desired contradiction, we are focused on accepting colored computation paths of a colored automaton and see how the computation paths can turn into different colors if we alter certain portions of input strings.
The proof will further exploit special properties of such a colored automaton by analyzing the behaviors of its time transition record of stack contents (called a \emph{stack history}) generated by this colored automaton. The detailed combinatorial analysis of the stack history will be presented in Sections \ref{sec:case-one}--\ref{sec:case-two}. The analysis itself is interesting on its own right. The proof of the main theorem will be split into two cases. In Case 1, the proof is supported by two key statements, Proposition \ref{size-stack-content} and Proposition \ref{symbol-length-small} (for a special case, Proposition  \ref{stack-height-bound}), in which we estimate the height of stack contents at certain points of a stack history.
The proofs of these propositions are quite contrive to some extent. These estimations provide two contradictory upper and lower bounds of the height, leading to the desired contradiction. In Case 2, we transform this case back to Case 1 by constructing a ``reversed'' colored automaton in Proposition \ref{reversing-machine} in Section \ref{sec:case-two}.

\s

We expect that  colored automata may find useful applications to other issues arising in  automata theory and we strongly hope that our analysis of stack history may shed another insight into the behaviors of other intriguing  automata models.

%%%%%%%%%%%%%%%%%%%%%%%%%
%%%%%%%%%%%%%%%%%%%%%%%%%
\section{Preparation for the Proof}\label{sec:preliminaries}

Before giving the awaiting proof of the main theorem (Theorem \ref{UCFL2V-by-CFLSV}) in Section \ref{sec:main-theorem}, we  wish to explain key notions and notation necessary to read through the rest of this paper.

Let $\nat$ denote the set of all \emph{natural numbers} (i.e., nonnegative integers) and define  $\nat^{+}=\nat-\{0\}$. Given two integers $m$ and $n$ with $m\leq n$, the notation $[m,n]_{\integer}$ denotes the \emph{integer interval} $\{m,m+1,m+2,\ldots,n\}$. When $n\geq1$, we further abbreviate $[1,n]_{\integer}$ as $[n]$. All \emph{logarithms} are taken to the base $2$ unless otherwise stated.
Given a finite set $A$, $\PP(A)$ denotes the {\em power set} of $A$ (i.e., the collection of all subsets of $A$).
The notation $|A|$ for a finite set $A$ refers to its {\em cardinality} (i.e., the number of all distinct elements in $A$).

An {\em alphabet} is a finite nonempty set of ``symbols'' or ``letters.''
Given such an alphabet $\Sigma$, a {\em string $x$ over $\Sigma$} is a finite series of symbols taken from $\Sigma$ and the {\em length} (or {\em size}) of $x$, denoted by $|x|$, is the total number of symbols in $x$.
We use $\lambda$ to express the {\em empty string} of length $0$.
The set of all strings over $\Sigma$ is denoted by $\Sigma^*$ and a {\em language over $\Sigma$} is a subset of $\Sigma^*$. Let $\Sigma^{+} = \Sigma^*-\{\lambda\}$.  The set $\Sigma^n$ for a number $n\in\nat$ is composed of all strings of length $n$ and the set $\Sigma^{\geq k}$ (resp., $\Sigma^{\leq k}$) consists of all strings of length at least $k$ (resp., at most $k$). For a symbol $a$ and a language $A$, the notation $aA$ stands for the set $\{ax\mid x\in A\}$.
Given two strings $x$ and $y$ over the same alphabet, the notation  $x\sqsubseteq y$ indicates that $x$ is a {\em substring} of $y$; namely,  $y$ equals $uxv$ for certain two strings $u$ and $v$.  Moreover, for a string $x$ and an index $i\in[|x|]$,  $(x)_i$ expresses a unique substring made up only of the first $i$ symbols of $x$.
Such a string is also called a \emph{prefix}. For example, $(0100)_1 = 0$ and $(0100)_3 = 010$. Clearly, $(x)_i\sqsubseteq x$ and $(x)_{|x|} = x$ hold.
Given a string $x=x_1x_2\cdots x_{n-1}x_n$ with $x_i\in\Sigma$ for all $i\in[n]$, the \emph{reversal} of $x$, denoted by $x^R$, is the string $x_nx_{n-1}\cdots x_2x_1$.

A \emph{multi-valued partial function} generally maps elements of a given set to subsets of another (possibly the same) set. Slightly different from a conventional notation\footnote{Another expression $f:A\to B$ is customarily used in computational complexity to express a multi-valued partial function.}  (\eg \cite{Sel94,Sel96}), we write $f:A\rightarrow \PP(B)$ for two sets $A$ and $B$ to refer to a multi-valued partial function that takes an element in $A$ as input and produces a certain number of elements in $B$.
In particular, when $f(x)=\setempty$, we conventionally say that $f(x)$ is {\em undefined}. The {\em domain} of $f$, denoted by $\dom(f)$, is therefore the set $\{x\in A\mid f(x)\text{ is not undefined }\}$. Given a constant $k\in\nat^{+}$,
$f$ is said to be {\em $k$-valued} if $|f(x)|\leq k$ holds for every input $x$ in $A$.
For two multi-valued partial functions $f,g:A\rightarrow\PP(B)$, we say that $g$ is a {\em refinement} of $f$ (or $f$ is {\em refined by} $g$), denoted by $f\sqsubseteq_{ref}g$, if (i) $\dom(f)=\dom(g)$ and (ii)  $g(x)\subseteq f(x)$ (set inclusion) holds for every $x\in\dom(g)$ \cite{Sel94}. For any two function classes $\FF$ and $\GG$, the succinct notation $\FF\sqsubseteq_{ref}\GG$ is used when every function in $\FF$ has a refinement in $\GG$.

Our mechanical model of computation is a {\em one-way nondeterministic pushdown automaton} (or an {\em npda}, for short) with/without a write-only output tape, allowing \emph{$\lambda$-moves} (or \emph{$\lambda$-transitions}).
We use an infinite input tape, which holds two special endmarkers: the left endmarker $\cent$ and the right endmarker $\dollar$. Let  $\check{\Sigma}$ stand for the set $\Sigma\cup\{\cent,\dollar\}$.
In addition, we use a semi-infinite output tape, on which its tape head is initially positioned at the first (\ie the leftmost) tape cell and moves only in one direction (to the right) whenever it writes a non-blank symbol.
Formally, an npda $M$ with an output tape is a tuple $(Q,\Sigma,\{\cent,\dollar\},\Gamma,\Theta,\delta,q_0,\bot, Q_{acc},Q_{rej})$ with a finite set $Q$ of inner states, an input alphabet $\Sigma$, a stack alphabet $\Gamma$, an output alphabet $\Theta$, the initial state $q_0\in Q$, the bottom marker $\bot \in\Gamma$, a set $Q_{acc}$ (resp., $Q_{rej}$) of accepting (resp., rejecting) states satisfying $Q_{acc}\cup Q_{rej} \subseteq Q$, and a transition function $\delta:(Q-Q_{halt})\times (\check{\Sigma}\cup \{\lambda\}) \times \Gamma \rightarrow \PP(Q\times \Gamma^*\times (\Theta\cup\{\lambda\}))$, where $Q_{halt}=Q_{acc}\cup Q_{rej}$.
The input tape is indexed by natural numbers with $\cent$ in the $0$th cell. When an input $x=x_1x_2\cdots x_n$ of length $n$ is given, it is placed in cells indexed from $1$ to $n$, where $\dollar$ is at the $(n+1)$th cell. A stack holds a series $s_ks_{k-1}\cdots s_{1}s_{0}$ of stack symbols in such a way that $s_0=\bot $ and $s_k$ is the topmost symbol.
We demand that $M$ should neither
remove $\bot $ nor replace it with any other symbol at any step; that is, for any tuple $(p,q,s,\sigma,\tau)$,
$(p,s,\tau)\notin \delta(q,\sigma,\bot)$ holds if $s$ does not contain $\bot $ at its bottom.
Conventionally, we say that the stack is \emph{empty} if it contains only the bottom marker $\bot $. Moreover, $M$ is not allowed to use $\bot $ as an ordinary stack symbol; that is, $(p,s,\tau)\notin \delta(q,\sigma,\gamma)$ if $\gamma\neq\bot$ and $\tau$ contains $\bot$.
The output tape must be  {\em write-only}; namely,
whenever $M$ writes a non-blank symbol on this tape,
its tape head must move to the right.
It is important to recognize two types of $\lambda$-moves.
When $\delta$ is applied to tuple $(q,\lambda,\gamma)$, $M$ modifies the current
contents of its stack and its output tape while neither scanning input symbols nor moving its input-tape head. In contrast,  when $(p,w,\lambda)\in \delta(q,\sigma,\gamma)$ holds, $M$ neither moves its output-tape head nor writes any non-blank symbol onto the output tape.

A \emph{configuration} of $M$ on input $x$ is a triplet $(q,i,w)$, in which $M$ is in inner state $q$ (with $q\in Q$), its tape head scans the $i$th cell (with $i\in[0,|x|+1]_{\integer}$), and its stack contains $w$ (with $w\in\Gamma^*$). The \emph{initial configuration} is $(q_0,0,\bot )$ and an accepting (resp., a rejecting) configuration is a configuration with an accepting state (resp., a rejecting state). A \emph{halting configuration} is either an accepting or a rejecting configuration. A \emph{computation path} of $M$ on $x$ is a series of configurations of $M$ on $x$, starting with the initial configuration, for which any non-initial configuration in the series must be reached from its predecessor by a single application of $\delta$.

Whenever we need to discuss an \emph{npda having no output tape}, we  drop ``$\Theta$'' as well as ``$\Theta\cup\{\lambda\}$'' from the aforementioned definition of $M$ and $\delta$. As stated in Section \ref{sec:introduction}, we consider only npda's whose computation paths all  terminate within $O(n)$ steps, where $n$ refers to any input size, and this particular condition concerning the termination of computation is conventionally  called the {\em termination condition} \cite{Yam14a}.  Throughout this paper, all npda's are implicitly assumed to satisfy this termination condition.

In general, an {\em output} (\emph{outcome} or \emph{output string}) of $M$ along a given  computation path refers to a string over $\Theta$ written down on the output tape when the computation path terminates. Such an output is called
{\em valid} (or {\em legitimate}) if the corresponding computation path is an accepting computation path (\ie $M$ enters an accepting state along this computation path).
Given a function $f$, we say that an npda $M$ with an output tape {\em computes} $f$ if, on  every input $x$, $M$ produces exactly all the strings in $f(x)$ as valid outputs; namely, for every pair $(x,y)$, $y\in f(x)$ if and only if $y$ is a valid outcome of $M$ on the input $x$.
Notice that an npda can generally produce more than one valid output string, its computed function inherently becomes multi-valued.
Because invalid outputs produced by $M$ are all discarded from our arguments in the subsequent sections, we will refer to valid outputs as just ``outputs'' unless otherwise stated.

The notation $\cflmv$ (resp., $\cflkv{k}$ for a fixed constant $k\in\nat^{+}$) stands for the class of all multi-valued (resp., $k$-valued) partial functions that can be computed by appropriate npda's with write-only output tapes in linear time. When $k=1$, in particular, we customarily write $\cflsv$ instead of $\cflkv{1}$.
In addition, we define $\ucflkv{k}$ as the collection of all functions $f$ in $\cflkv{k}$ for which an appropriate npda $M$ equipped with an output tape computes $f$ with the  extra condition (called the {\em unambiguous computation condition})  that, for every input $x$ and every value $y$ in $f(x)$, there exists exactly one accepting computation path of $M$ on $x$ producing $y$.
It follows by their definitions that $\ucflkv{k}\subseteq \cflkv{k} \subseteq\cflmv$. Since any function producing exactly $k+1$ values cannot belong to $\cflkv{k}$ by definition, $\cflkv{k}\neq\cflkv{(k+1)}$ holds; therefore, in particular, we obtain $\cflsv\neq\cflmv$. Notice that this inequality does not directly lead to the desired conclusion $\cflmv\not\sqsubseteq_{ref}\cflsv$.

To describe behaviors of an npda's stack, we closely follow terminology from  \cite{Yam08,Yam16}.
A {\em stack content} is formally a series $z_mz_{m-1}\cdots z_1z_0$ of stack symbols sequentially stored into a stack (in our convention,  $z_0$ is the bottom marker $\bot $ and $z_m$ is a symbol at the top of the stack).
A {\em stack content at the $i$th cell position} refers to a stack content obtained just after the tape head scans and then moves off the $i$th cell of the input tape. A series of stack contents produced along a computation path is briefly referred to as a \emph{stack history}.

%%%%%%%%%%%%%%%%%%%%%%%%%%
%%%%%%%%%%%%%%%%%%%%%%%%%%
\section{Proof of the Main Theorem}\label{sec:main-theorem}

Our ultimate goal is to solve negatively a question that was posed in \cite{KSY07} and reformulated in \cite{Yam14b} as in the form of Question \ref{CFLMV-refine-k=1}. For this purpose, we intend to prove the main theorem (Theorem \ref{UCFL2V-by-CFLSV}).
As an example of a  non-refinable function that witnesses the theorem, we will present a special function, called $h_3$, which belongs to $\ucfltwov$ (shown in Section \ref{sec:example-function}), and then give an explanation of why no refinement of this function is found in $\cflsv$, resulting in the main theorem, namely, $\ucfltwov\not\sqsubseteq_{ref}\cflsv$. To simplify our proof, we will introduce in Section \ref{sec:colored-automata} a computational model of \emph{colored automata}, which have no output tapes but have ``colors'' to specify their outcomes.  We will conduct a combinatorial analysis on a stack history of such colored automata in Section \ref{sec:case-one}--\ref{sec:case-two}.

%%%%%
\subsection{An Example Function}\label{sec:example-function}

Our example function $h_3$ is a natural extension of a well-known deterministic context-free language $\{x\#x^R\mid x\in\{0,1\}^*\}$ (marked even-length palindromes), where $\#$ is a distinguished symbol not in $\{0,1\}$, used as a separator.
Let us define two supporting languages
$L=\{x_1\# x_2\# x_3\mid x_1,x_2,x_3\in\{0,1\}^*\}$ and $L_3=\{w \mid \exists x_1,x_2,x_3[ w=x_1\# x_2\# x_3\in L, \exists (i,j)\in I_3\;[x_i^R=x_j]] \}$, where
$I_3=\{(i,j)\mid i,j\in\nat^{+}, 1\leq i<j\leq 3\}$.
We then introduce the desired function $h_3$ by setting $h_3(w)=\{ 0^i1^j\mid (i,j)\in I_3, x_i^{R}=x_j \}$ if $w=x_1\# x_2\# x_3\in L$, and $h_3(w)=\setempty$ otherwise.
It thus follows that
$L_3=\{w\in L\mid h_3(w)\neq\setempty\}$.
As simple examples, if $w$ has the form $x\#x^R\#y$ with $x\neq y$, then $h_3(w)$ equals $\{011\}$; in contrast, if $w=x\#x^R\#x$, then $h_3(w)$ is $\{011,00111\}$.

Let us verify the following proposition.

\begin{proposition}
The above function $h_3$ is in $\ucfltwov$.
\end{proposition}

\begin{proof}
Obviously, the function $h_3$ is 2-valued.
Targeting $h_3$, let us consider the following npda $M$ equipped with a write-only output tape. On any input $w$, $M$ deterministically checks whether $w$ is of the form $x_1\#x_2\#x_3$ in $L$ by moving its input-tape head from the left to the right by counting the number of $\#$ in $w$. At the same time, $M$ guesses (i.e., nondeterministically chooses) a pair $(i,j)\in I_3$, writes $0^i1^j$ onto its  output tape, stores $x_i$ into a stack, and then checks whether $x_i^R$ matches $x_j$ by retrieving $x_i$ in a reverse order from the stack. If $x_i^R=x_j$ holds, then $M$ enters an accepting state; otherwise, it enters a rejecting state.

To be more formal, the desired npda $M=(Q,\Sigma,\{\cent,\dollar\},\Gamma,\Theta,\delta,q_0,\bot, Q_{acc},Q_{rej})$ is defined as follows. Let $\Sigma=\{0,1,\#\}$, $\Gamma=\{0,1,\bot \}$, and $\Theta=\{0,1\}$. Moreover, let $Q_{acc}=\{q_{acc}\}$ and $Q_{rej}=\{q_{rej}\}$. The transition function $\delta$ consists of the following transitions. The first move of $M$ is a nondeterministic move of $\delta(q_0,\cent,\bot )=\{(q_{12}^{(0)},\bot,01^2), (q_{23}^{(0)},\bot, 0^21^3), (q_{13}^{(0)},\bot, 01^3)\}$. With respect to $q_{12}^{(0)}$, this first step is followed by a series of transitions:  $\delta(q_{12}^{(0)},\sigma,a) =\{(q_{12}^{(0)},\sigma a,\lambda)\}$,  $\delta(q_{12}^{(0)},\#,a) =\{(q_{12}^{(1)},a,\lambda)\}$, $\delta(q_{12}^{(1)},\sigma,\sigma) =\{(q_{12}^{(1)},\lambda,\lambda)\}$, $\delta(q_{12}^{(1)},\#,\bot ) =\{(q_{12}^{(2)},\bot, \lambda)\}$, $\delta(q_{12}^{(2)},\sigma,\bot ) =\{(q_{12}^{(2)},\bot, \lambda)\}$, and $\delta(q_{12}^{(2)},\dollar,\bot ) =\{(q_{acc},\bot, \lambda)\}$, where $\sigma\in\Sigma$ and $a\in\Gamma$. For all other transitions, $M$ changes its inner states to $q_{rej}$.  For the other two states $q_{23}^{(0)}$ and $q_{13}^{(0)}$, we can define similar sets of transitions and thus we omit their precise descriptions.

It therefore follows by the above definition that, for each choice of $(i,j)$ in $I_3$, there is at most one accepting computation path producing $0^i1^j$.
It is not difficult to verify that $M$ correctly computes $h_3$. Therefore, $h_3$ belongs to $\ucfltwov$.
\end{proof}

We have obtained the example function $h_3$, which belongs to $\ucfltwov$. To complete the proof of the main theorem, it therefore suffices to
verify the following proposition regarding the non-existence of a refinement of the function $h_3$.

\begin{proposition}\label{h3-no-refinement}
The function $h_3$ has no refinement in $\cflsv$.
\end{proposition}

In the subsequent five subsections, we will describe the proof of Proposition \ref{h3-no-refinement} and thus derive the main theorem.

%%%%
%%%%
\subsection{Colored Automata}\label{sec:colored-automata}

Our proof of Proposition \ref{h3-no-refinement} proceeds by way of contradiction. To lead to the desired contradiction, we first assume that $h_3$ has a  refinement, say,  $g$ in $\cflsv$. Since $g$ is single-valued, in what follows, we rather write $g(x)=y$ instead of $g(x)=\{y\}$ for $x\in\dom(f)$.
Take an npda $N$ computing $g$ with a write-only output tape. Assume that  $N$ has
the form  $(Q,\Sigma,\{\cent,\dollar\},\Gamma,\Theta,\delta,q_0,\bot, Q_{acc},Q_{rej})$ with a transition function $\delta:(Q-Q_{halt})\times (\check{\Sigma}\cup\{\lambda\}) \times \Gamma\rightarrow \PP(Q\times \Gamma^*\times(\Theta\cup\{\lambda\}))$. Notice that $\Sigma=\Theta=\{0,1\}$ by the definition of $g$.

Unfortunately, we find it difficult to directly follow and analyze the moves of $N$'s output-tape head. To overcome this difficulty, we try to modify $N$ into a new variant of {\em npda having no output tape}.
As seen later, this modification is possible because $g$'s output values are limited only to strings of \emph{constant lengths}.
Now, let us introduce this new machine, dubbed as ``colored'' automaton, which has no output tapes but uses  ``colored'' stack symbols. Using a finite set $C$ of ``colors,'' a {\em colored automaton}   $M=(Q',\Sigma,\{\cent,\dollar\},\Gamma',C,\delta',q'_0,\bot, Q'_{acc},Q'_{rej})$  partitions its stack alphabet $\Gamma'$, except for the bottom marker $\bot$, into sets $\{\Gamma'_{\xi}\}_{\xi\in C}$; namely, $\bigcup_{\xi\in C}\Gamma'_{\xi}  = \Gamma'-\{\bot \}$ and $\Gamma'_{\xi}\cap \Gamma'_{\xi'}=\setempty$ for any distinct pair $\xi,\xi'\in C$. Let $\tilde{\Gamma}_{\xi} = \Gamma'_{\xi}\cup\{\bot \}$ for each color $\xi\in C$.
For a {\em color} of stack symbol $\gamma$, we say that $\gamma$ is \emph{in color $\xi$} if $\gamma$ is in $\tilde{\Gamma}_{\xi}$.  Note that $\bot $  has all colors.

The transition function $\delta'$ maps $(Q'-Q'_{hatl}) \times (\check{\Sigma}\cup\{\lambda\}) \times \Gamma'$ to $\PP(Q'\times (\Gamma')^*)$. Given a substring $z$ of input $x$, we say that $M$ \emph{read off} $z$ if $M$ starts with reading the leftmost symbol of $z$, continues reading the entire symbols of $z$, makes all possible $\lambda$-moves after reading an input symbol, and moves its tape head out of $z$ after processing the rightmost symbol of $z$. The notions of computation, computation path, and accepting/rejecting computation path are defined similarly to the case of npda's.

Given a color $\xi\in C$, we call a computation path of $M$ a {\em $\xi$-computation path} if all  configurations along this computation path use only stack symbols in color $\xi$.
In this case, such a computation path is also said to be \emph{colored} (in color $\xi$). Most importantly, we demand that all computation paths of $M$ should be colored.
An {\em output} of $M$ on input $x$ is composed of all colors $\xi$ in $C$ for which there is an accepting $\xi$-computation path of $M$ on $x$.

Let us verify that the aforementioned CFLSV function $g$ can be computed by an appropriately chosen colored automaton.

\begin{lemma}\label{construction}
Assuming $g\in\cflsv$, there exists a colored automaton $M$ that computes $g$.
\end{lemma}

\begin{proof}
Associated with the set $I_3$ introduced in Section \ref{sec:example-function}, we define a new set $\bar{I}_3=\{0^i1^j\mid (i,j)\in I_3\}$ and another set $\bar{I}^{part}_3$ composed of all substrings of any string in $\bar{I}_3$.
Notice that $\lambda\in \bar{I}^{part}$.
Let us recall the aforementioned npda $N =  (Q,\Sigma,\{\cent,\dollar\},\Gamma,\Theta,\delta,q_0,\bot,  Q_{acc},Q_{rej})$ that computes $g$ with its  write-only output tape.
Now, we wish to construct a new colored automaton $M=(Q',\Sigma,\{\cent,\dollar\},\Gamma',\bar{I}_3, \delta',q'_0,\bot, Q'_{acc},Q'_{rej})$ that simulates $N$.

We start by setting $Q'= Q\times (\bar{I}_3\cup\{\lambda\})\times \bar{I}^{part}_3$ and $q'_0 = (q_0,\lambda,\lambda)$. In addition,  let $\Gamma'=\{\bot \}\cup \{\tau^{(t)}\mid \tau\in \Gamma-\{\bot \},t\in\bar{I}_3\} \cup \bar{I}_3$.
Intuitively, taking an input $x$, $M$ first {\em guesses} (\ie nondeterministically chooses) an output string $t$ of $g(x)$. Note that, at the first step of $M$ on the input $x$, $M$ pushes $t$ into its stack as a stack symbol in $\Gamma'$.
Whenever $N$ pushes $u$ to its stack along a specific computation path $\gamma$, $M$ pushes the corresponding color-$t$ symbol $u^{(t)}$ into its own stack. Further along this computation path $\gamma$, $M$ keeps using only color-$t$ stack symbols, which are marked by
the superscript ``$(t)$.''
Instead of having an output tape, $M$ remembers the string $t$ produced  on $N$'s output tape. This is possible because $t$ is of a constant length. Whenever $N$ enters an accepting state, say, $q_{acc}$ with an output string that matches the initially guessed string $t$ of $M$, $M$ enters an appropriate accepting state, say, $(q_{acc},t,t)$. In other cases, $M$ rejects the input immediately.

To realize this intuition, we formally define $M$'s transition function $\delta'$ based on $\delta$ as follows.
Assume that, at the first step, $N$ applies a transition of the form $(p,u\bot,\xi)\in \delta(q,\cent,\bot )$. The corresponding transition of $M$ is $((p,t,\xi),u^{(t)}\bot )\in \delta'(q'_0,\cent,\bot )$ for all $t \in \bar{I}_3$.
If $N$ makes a transition of the form $(p,w,\xi)\in \delta(q,\sigma,\gamma)$ with $\sigma\neq\cent$, then  $M$ applies a transition $((p,t,\tau\xi),w^{(t)})\in \delta'((q,t,\tau),\sigma,\gamma^{(t)})$ as long as $M$ is in inner state $(q,t,\tau)$, where $w^{(t)}$ is defined recursively to be  $u^{(t)}$ if $w=u\in\Gamma$, and
$u^{(t)}_1u^{(t)}_2$ if $w=u_1u_2$. In the end of computation, assume that $N$'s transition is of the form $(p,w,\xi)\in\delta(q,\dollar,\gamma)$ with $p\in Q_{halt}$. If $p\in Q_{acc}$ and $t=\tau\xi$, then we set $((p,t,t),w^{(t)})\in\delta'((q,t,\tau),\dollar,\gamma^{(t)})$.
In contrast, if $p\in Q_{rej}$, then we set $((p,t,\tau\xi),w^{(t)})\in\delta'((q,t,\tau),\dollar,\gamma^{(t)})$. However, if $p\in Q_{acc}$ but $t\neq \tau\xi$, then we use the new symbol $q'_{rej}$ and set $((q'_{rej},t,\tau\xi),w^{(t)})\in \delta'((q,t,\tau),\dollar,\gamma^{(t)}))$. Finally, we define $Q'_{acc} =\{ (q,t,t)\mid q\in Q_{acc},t\in \bar{I}_3\}$ and  $Q'_{rej} =\{(q,t,s)\mid t\neq s, q\in Q_{rej}\cup\{q'_{rej}\},t\in \bar{I}_3, s\in \bar{I}_3^{part}\}$.
\end{proof}

To simplify notation in our argument, we describe the colored automaton $M$ guaranteed by Lemma \ref{construction} as
$(Q,\Sigma,\{\cent,\dollar\},\Gamma,I_3,\delta,q_0,\bot, Q_{acc},Q_{rej})$. Notice that we consciously use $I_3$ instead of $\bar{I}_3$. For the subsequent analysis of the behaviors of $M$, it is also useful to restrict the behavior of $M$. A colored automaton $M$ is said to be {\em in an almost ideal shape} if $M$ satisfies all of the following six conditions.

\renewcommand{\labelitemi}{$\circ$}
\begin{enumerate}\vs{-2}
  \setlength{\topsep}{-2mm}%
  \setlength{\itemsep}{0mm}%
  \setlength{\parskip}{0cm}%

\item There are only one accepting state $q_{acc}$ and one rejecting state $q_{rej}$. Moreover, the set $Q$ of inner states equals $\{q_0,q,q_{acc},q_{rej}\}$. The machine is always in state $q$ during its computation except for the initial and final configurations.

\item An input-tape head  always moves to the right until it reaches $\dollar$.

\item The machine never aborts its computation; that is, $\delta$ is a \emph{total function} (\ie $\delta(q,\sigma,\gamma)\neq\setempty$ holds for any $(q,\sigma,\gamma)\in (Q-Q_{halt})\times \check{\Sigma} \times \Gamma$).

\item The machine never enters any halting  state before scanning the right endmarker $\dollar$.

\item As each stack operation, the machine (i) pops the topmost stack symbol, (ii) replaces the topmost stack symbol by another single stack symbol, or (iii) pushes extra one symbol onto the top of the stack after (possibly)   altering the then-topmost symbol; that is,  the range of $\delta$ must be of the form  $\PP(Q\times \Gamma^{\leq 2} \times  (\Theta\cup\{\lambda\}))$, where $\Gamma^{\leq k}$ is the set $\{\gamma\in\Gamma^*\mid |\gamma|\leq k\}$.

\item The stack never becomes empty  at any step of the computation except for the initial and the final configurations. In addition, at the first step of reading $\cent$, the machine  must push a stack symbol onto $\bot $ and this stack symbol  determines the stack color in the rest of its computation path. Before entering any halting state, the stack must become empty.
\end{enumerate}

It is well-known that, for any  context-free language $L$, there always exists an
npda (with no output tape) in an almost ideal shape that recognizes $L$ (see, \eg \cite{HU79}). Similarly, we can assert the following statement for colored automata.

\begin{lemma}[Almost Ideal Shape Lemma]\label{modify-ideal-shape}
Given any colored automaton, there is always another colored automaton in an almost ideal shape that produces the same set of output values.
\end{lemma}

For readability, we place the proof of Lemma \ref{modify-ideal-shape} in Appendix.

In the rest of this section, we fix the colored automaton $M$, which computes $g$ correctly.
We further assume  by Lemma \ref{modify-ideal-shape} that $M$ is in an almost ideal shape.

%%%%%

\begin{figure}[t]
\centering
%\begin{center}
%\includegraphics*[width=12.0cm]{stack-figure-01.pdf}
\includegraphics*[height=4.0cm]{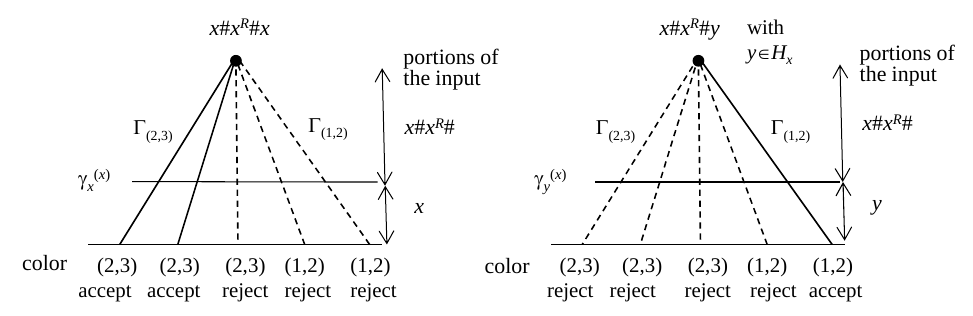}
\caption{Two computation trees of the colored automaton $M$ on two inputs $x\#x^Rx$ and $x\#x^R\#y$ with $y\in H_x$. The color of each computation path and its acceptance/rejection are written in the bottom of the figures.}\label{fig:colored-automata}
%\end{center}
\end{figure}

%%%%%

Hereafter, let us focus on inputs of the form $x\#x^R\#y$ for $x,y\in\{0,1\}^*$. For any string $x\in\{0,1\}^*$, we abbreviate  the set $\{y\in\{0,1\}^{|x|}\mid y\not\in\{x,x^R\}\}$ as $H_x$.
Given a number  $n\in\nat^{+}$,  $D^{(n)}_{(i,j)}$ denotes the set of all strings $x\in\{0,1\}^n$ for which there exists an accepting $(i,j)$-computation path of $M$ on input $x\#x^R\#x$. Since $g$ is single-valued, it follows that $D^{(n)}_{(1,2)}\cup D^{(n)}_{(2,3)}=\{0,1\}^n$. We therefore conclude that,  for every length $n$,  either $|D^{(n)}_{(1,2)}|\geq2^n/2$ or $|D^{(n)}_{(2,3)}|\geq 2^n/2$ (or both) holds.
We will discuss  in Sections \ref{sec:case-one}--\ref{sec:stack-size} the case where  $|D_{(2,3)}^{(n)}|\geq 2^n/2$ holds for infinitely many $n$'s and consider  in Section \ref{sec:case-two} the case where  $|D_{(2,3)}^{(n)}|\geq 2^n/2$ holds only finitely many $n$'s (thus, $|D_{(1,2)}^{(n)}|\geq 2^n/2$ holds for infinitely many $n$'s). To complete the proof of Proposition \ref{h3-no-refinement}, it suffices for us to obtain contradictions in both cases.

%%%%%
%%%%%
\subsection{Case 1: $D^{(n)}_{(2,3)}$ is Large for Infinitely Many  Lengths $n$}\label{sec:case-one}

Let us first consider the case where the relation
$|D^{(n)}_{(2,3)}|\geq 2^{n}/2$ holds for infinitely many lengths $n\in\nat$.
Take an arbitrary number $n\in\nat$ that is significantly larger than $3^{|Q|+|\Sigma|+|\Gamma|}$ and also satisfies
$|D^{(n)}_{(2,3)}|\geq 2^{n}/2$. We fix such a number $n$ throughout our
proof and we thus tend to drop script ``$n$'' whenever its omission is clear from the context; for instance, we intend to write $D_{(2,3)}$ instead of $D^{(n)}_{(2,3)}$.

By the property of the colored automaton $M$ computing $g$, it follows that, for any pair $x,y\in\{0,1\}^n$, if $y\in H_x$ (i.e., $y\notin\{x,x^R\}$),  then, for the input $x\#x^R\#y$, there always exists a certain accepting $(1,2)$-computation path of $M$. See Figure \ref{fig:colored-automata}.  However, since $g$ is single-valued,  there must be no \emph{accepting} $(1,2)$-computation path of $M$ on the input $x\#x^R\#x$ for every $x$ in $D_{(2,3)}$. In addition, no accepting $(1,2)$-computation path exists on all inputs of the form $x\#z\#y$ if $z\neq x^R$.
Since there could be a large number of accepting $(1,2)$-computation paths of $M$ on $x\#x^R\#y$, we need to choose one of them arbitrarily and take a close look at this particular computation path.

For the aforementioned purpose, we denote by $PATH_n$ the set of all possible \emph{accepting} $(1,2)$-computation paths of $M$ on inputs of the form $x\#x^R\#y$ for any two strings $x,y\in\{0,1\}^n$, and we fix a {\em partial assignment}  $\pi: D_{(2,3)} \times\{0,1\}^n \rightarrow PATH_n$ that, for any element $(x,y)$, if $y\in H_{x}$, then $\pi$ picks one of the accepting $(1,2)$-computation paths of $M$ on the input $x\#x^R\#y$; otherwise, let $\pi(x,y)$ be \emph{undefined}, for simplicity. Hereafter, we abbreviate $\pi(x,y)$ as $p_{x,y}$. Thus, $p_{x,y}$ is uniquely determined from $(x,y)$ whenever $\pi(x,y)$ is defined.

Along the unique accepting $(1,2)$-computation path $p_{x,y}$ of $M$ on each  input $x\#x^R\#y$, the notation $\gamma^{(x)}_{i,y}$ is used to denote a stack content obtained by $M$ just after reading off the first $i$ symbols of $x\#x^R\#y$ (making all possible $\lambda$-moves).
Note that, for each $x\in D_{(2,3)}$ and any  $y\in H_{x}$, along the  accepting $(1,2)$-computation path $p_{x,y}$ associated with the input $x\#x^R\#y$,
$M$ produces unique stack contents $\gamma^{(x)}_{|x\#|,y}$ and $\gamma^{(x)}_{|x\# x^R\#|,y}$.
For convenience,  we abbreviate as $\gamma_{y}^{(x)}$  the stack content $\gamma^{(x)}_{|x\#x^R\#|,y}$, which is produced by $M$ just after reading the substring $x\#x^R\#$ of the input $x\#x^R\#y$.

In Sections \ref{sec:fundamental-prop}--\ref{sec:case-two}, we plan to evaluate how many strings $x$ in $D_{(2,3)}$ satisfy each of the following conditions.

\renewcommand{\labelitemi}{$\circ$}
\begin{enumerate}\vs{-1}
  \setlength{\topsep}{-2mm}%
  \setlength{\itemsep}{0mm}%
  \setlength{\parskip}{0cm}%

\item Strings $x$ make $\gamma_y^{(x)}$ small in size for all strings $y$ in $H_x$.

\item Strings $x$  make $\gamma_y^{(x)}$  relatively large in size for certain strings $y$ in $H_x$.
\end{enumerate}

Proposition \ref{size-stack-content} gives a lower bound on the number of strings $x$ satisfying  Condition (1), whereas Propositions \ref{stack-height-bound} and \ref{symbol-length-small} provide lower bounds on the number of strings $x$ for Condition (2). If these two bounds are large enough, then they guarantee the existence of a string that meet both conditions, clearly leading to a contradiction. Therefore, we conclude that $M$ cannot exist, closing Case 1.
Proposition \ref{size-stack-content}  will be verified in Section \ref{sec:fundamental-prop} and Propositions \ref{stack-height-bound} and \ref{symbol-length-small} will appear in Section \ref{sec:stack-size}.

%%%%%
%%%%%
\subsection{Fundamental Properties of a Stack History}\label{sec:fundamental-prop}

A key to our proof is an analysis of a stack history of the given colored automaton $M$.
In the following series of lemmas and a proposition, we will explore  fundamental properties of a stack history of $M$ along accepting $(1,2)$-computation path $p_{x,y}$ on any input of the form $x\# x^R\# y$. Those properties are essential in dealing with Case 1.
We start with showing a simple property asserting that, in the above stack history, the same stack content does not appear twice or more.

\begin{lemma}\label{two-stack-difference}
Fix $x\in D_{(2,3)}$ and $y\in\{0,1\}^n$ arbitrarily. For any accepting $(1,2)$-computation path $p_{x,y}$ of $M$ on the input $x\#x^R\#y$, there is no pair $(i_1,i_2)$ of cell positions satisfying that $|x|<i_1<i_2\leq |x\#x^R\#|$ and $\gamma^{(x)}_{i_1,y} =  \gamma^{(x)}_{i_2,y}$. Moreover, the same statement is true when $1\leq i_1<i_2\leq |x|$.
\end{lemma}

\begin{yproof}
Assume that $M$ has an accepting $(1,2)$-computation path $p_{x,y}$ for the given strings $x\in D_{(2,3)}$ and $y\in\{0,1\}^n$.
If the first part of the lemma fails, then the inequality  $\gamma^{(x)}_{i_1,y}=\gamma^{(x)}_{i_2,y}$ must hold for a certain pair $(i_1,i_2)$ satisfying $|x|<i_1<i_2\leq |x\#x^R\#|$. We remove from $x\#x^R\# y$ all input symbols located in between cell positions $i_1+1$ and $i_2$ and then express the resulted string by $x\#x'\#y$. Since $\gamma^{(x)}_{i_1,y}=\gamma^{(x)}_{i_2,y}$,  we can obtain  a new $(1,2)$-computation path, along which  $M$ still enters an accepting state on this input $x\#x'\#y$. However, since  $x'\neq x^R$, there must be no accepting $(1,2)$-computation path on $x\#x'\#y$, a contradiction. The second part of the lemma follows by a similar argument.
\end{yproof}

Lemma \ref{two-stack-difference} can be generalized as follows.

\begin{lemma}\label{stack-difference-at-i}
Let $x_1,x_2,y_1,y_2\in\{0,1\}^n$ satisfy both $y_1\in H_{x_1}$ and $y_2\in H_{x_2}$ and let $i_1,i_2\in\nat$ satisfy
$1\leq i_1,i_2 \leq |x_1\#x_1^R\#|$.
Let $w_1=x_1\# x_1^R\# y_1$ and $w_2=x_2\# x_2^R\# y_2$.
Assume that one of the following three conditions holds: (i) $i_1\neq i_2$,  (ii) $1\leq i_1=i_2\leq |x_1\#|$
and $(w_1)_{i_1}\neq (w_2)_{i_2}$, and (iii) $|x_1\#|< i_1=i_2 \leq |x_1\#x_1^R\#|$ and $(w_1)_{i_1}\neq (w_2)_{i_2}$. If the two computation paths $p_{x_1,y_1}$ and $p_{x_2,y_2}$ exist, then  $\gamma^{(x_1)}_{i_1,y_1}\neq \gamma^{(x_2)}_{i_2,y_2}$ holds.
\end{lemma}

\begin{yproof}
Let $(x_1,x_2,y_1,y_2,i_1,i_2)$ be parameters given as in the premise of the lemma.
We consider the three cases (i)--(iii) separately.
The first case of $i_1\neq i_2$ comes from an argument similar to the proof
of Lemma \ref{two-stack-difference}. In what follows, we therefore assume that $i_1=i_2$ and denote them as $i$ for simplicity.

Next, we consider the second case where  $1\leq i \leq |x_1\#|$ and $(w_1)_i\neq (w_2)_i$.
Toward a contradiction, we assume that $\gamma^{(x_1)}_{i,y_1}=\gamma^{(x_2)}_{i,y_2}$. Assume that there are two accepting $(1,2)$-computation paths $p_{x_1,y_1}$ and $p_{x_2,y_2}$ generated by $M$ respectively on the inputs $w_1$ and $w_2$.
Notice that $i<|x_1|$ because, otherwise, $(w_1)_i=(w_2)_i$ follows.
Take a unique nonempty string $u_1$ satisfying $x_1=(x_1)_iu_1$.
Since $\gamma^{(x_1)}_{i,y_1}=\gamma^{(x_2)}_{i,y_2}$ by our assumption, we can swap the initial segments of these two computation paths restricted to the first $i$ steps, corresponding to  the first $i$ bits of the above two inputs. As a result, we obtain another accepting $(1,2)$-computation path on the input $(x_2)_iu_1\#x_1^R\#y_1$. Since $M$ precisely computes $g$,
$(x_2)_{i}u_1$ must equal $x_1$.
From this, $(w_1)_{i}=(w_2)_i$ follows instantly, a clear contradiction against $(w_1)_i\neq (w_2)_i$. Therefore, we obtain  $\gamma^{(x_1)}_{i,y_1}\neq \gamma^{(x_2)}_{i,y_2}$.

A similar argument can handle the third case, in which both $|x_1\#|<i \leq|x_1\#x_1^R\#|$ and $(x_1)_i\neq (x_2)_i$ hold.
Firstly, we claim that $x_1\neq x_2$ because, otherwise, $x_1\#x_1^R\#$ coincides with $x_2\#x_2^R\#$, contradicting $(w_1)_i\neq (w_2)_i$. Since $x_1\neq x_2$, we obtain $i<|x_1\#x_1^R|$. Take a unique string $u_1$ that satisfies $(x_1\#x_1^R)_iu_1 = x_1\# x_1^R$. Similarly to the second case, assuming $\gamma^{(x_1)}_{i,y_1}=\gamma^{(x_2)}_{i,y_2}$, there must exist an accepting $(1,2)$-computation path of $M$ on the input $(x_2\#x_2^R)_iu_1\#y_1$. From this,  we can draw a conclusion that $x_2^Ru_1 = x_1^R$; thus, $(w_1)_i=(w_2)_i$ follows, contradicting our assumption.
\end{yproof}

%%%%%

In what follows, we want to estimate the number of strings $x$ in $D_{(2,3)}$ for which their corresponding stack contents  $\gamma^{(x)}_{y}$ are  small in size for an \emph{arbitrary} string $y$ in $H_x$. In what follows, we will show its lower bound, which is sufficiently large.

\begin{proposition}\label{size-stack-content}
There exist two constants $d_1,d_2\in\nat^{+}$, independent of $(n,x,y)$, such that $|\{x\in D^{(n)}_{(2,3)}\mid \forall y\in H_{x}\,[ |\gamma^{(x)}_{y}|< d_1]\}| \geq |D^{(n)}_{(2,3)}| - d_2$.
\end{proposition}

Proposition \ref{size-stack-content} is one of the key statements necessary to handle Case 1. For the proof of this proposition, we need two supporting lemmas, Lemmas \ref{subpath-transform} and \ref{TF-bound}. To explain these lemmas, we need to introduce extra terminology and notation.

Let us recall from Section \ref{sec:colored-automata} that $\{\Gamma_t\}_{y\in I_3}$ is a ``color'' partition of $\Gamma$. Given two strings $u\in (\Gamma_{(1,2)})^{+}$ and $v\in (\Gamma_{(1,2)})^*$ and a string $z\in\{0,1\}^+$, we say that $M$ {\em transforms $u$ to $v$ while reading $z$ (along computation (sub)path $p$)} if $M$ behaves as follows along this particular computation (sub)path $p$:
(i) $M$ starts  in inner state $q$ with $uw\bot $ in its stack for a certain string $w\in(\Gamma^{(1,2)})^*$, scanning the leftmost input symbol of $z$, (ii) $M$ then reads all input symbols in $z$, including no endmarker, one by one, (iii) just after reading off $z$ (making all possible $\lambda$-moves), $M$ enters inner state $q$, leaving $vw\bot $ in the stack, and
(iv) $M$ does not access any symbol in $w$ while reading $z$.
The notation $TF_{M}(u,v)$ expresses the set of all strings of the form $z\#z'$ for any pair $z,z'\in\{0,1\}^*$ such that $M$ transforms $u$ to $v$ while reading $z\#z'$ along certain computation (sub)paths.

\begin{lemma}\label{subpath-transform}
Given any pair $u,v\in(\Gamma_{(1,2)})^*$, there exists at most one string $x'$ such that $x'$ is a substring of a certain string $x$ in $D_{(2,3)}$ and $M$ transforms $u$ to $v$ while reading $x'$ along an appropriate subpath of $p_{x,y}$ for a certain string $y$ in $\{0,1\}^n$.
\end{lemma}

\begin{yproof}
We prove the lemma by way of contradiction. Assume that there are two distinct strings $x_1,x_2\in\{0,1\}^*$ satisfying that $M$ transforms $u$ to $v$ while reading $x_1$ along computation subpath $p_1$ and $M$ transforms $u$ to $v$ while reading $x_2$ along computation subpath $p_2$.
Let us consider a string $x$ in $D_{(2,3)}$ such that $x$ contains $x_1$ as a substring and an accepting $(1,2)$-computation path $p_{x,y}$ on input $x\#x^R\#y$ contains $p_1$ as a subpath for a certain string $y\in\{0,1\}^n$.
Let $x'$ be a string obtained from $x$ by replacing $x_1$ with $x_2$. It is possible to swap the two subpaths $p_1$ and $p_2$ without changing the acceptance criteria of $M$. Therefore, $M$ must have an accepting $(1,2)$-computation path on the input $x'\#x^R\#y$. This is absurd since $x'\neq x$.
\end{yproof}

Next, we will  show a technical lemma, Lemma \ref{TF-bound}, which is essential to prove Proposition \ref{size-stack-content}. We already know from  Lemma \ref{two-stack-difference}  that  all elements in $\{\gamma^{(x)}_{i,y}\mid 1\leq i\leq|x\#x^R\#|\}$ are mutually distinct. Let us concentrate particularly on stack contents $\gamma^{(x)}_{i,y}$ of minimal size.
Given any pair $(x,y)$, we define $MSC_{x,y}$ (minimal stack contents) to be the collection of all stack contents $\gamma$ that meet the following requirement:
there exists a cell position $\ell$ with $|x\#|\leq \ell \leq |x\#x^R\#|$
such that
(i)  $\gamma = \gamma^{(x)}_{\ell,y}$ and  (ii) $|\gamma|\leq |\gamma^{(x)}_{\ell',y}|$ holds for any cell position $\ell'$ satisfying $|x\#|\leq \ell'\leq |x\#x^R\#|$. Condition (ii), in particular,  ensures that the size of $\gamma$ must be  minimum.
Note that  $MSC_{x,y}$ cannot be empty.

%%%
%%%
%\ignore{

\begin{figure}[t]
\centering
%\begin{center}
%\includegraphics*[width=12.0cm]{stack-figure-01.pdf}
\includegraphics*[height=4.5cm]{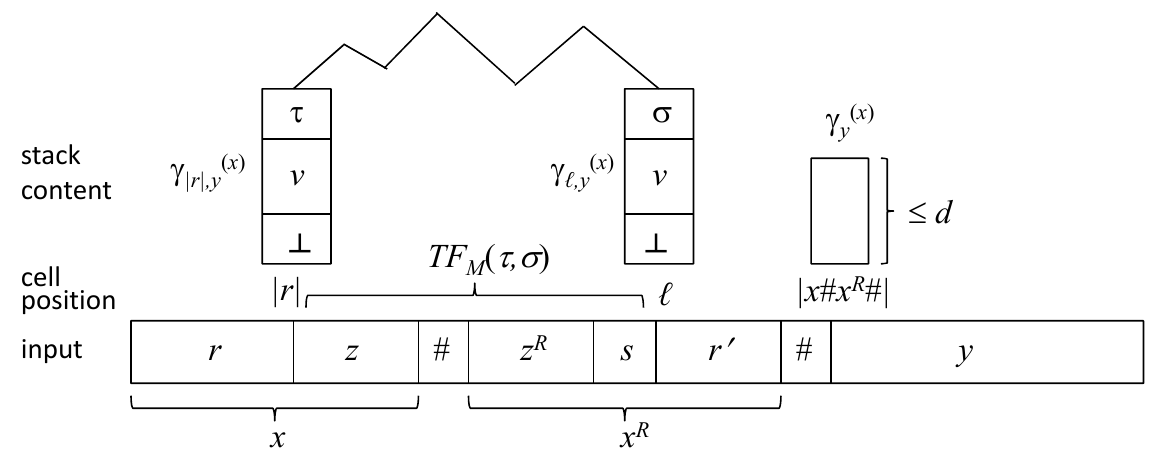}
\caption{A stack history stated in Lemma \ref{TF-bound}.}\label{fig:stack-history}
%\end{center}
\end{figure}

%\begin{figure}[t]
%\centering
%\includegraphics*[width=8.5cm]{figure-diagram-cropped.pdf}
%\caption{Inclusion/separation relationships among nonuniform state complexity classes}\label{fig:relationship}
%\end{figure}
%}
%%%
%%%

\begin{lemma}\label{TF-bound}
There exists a constant $d>0$, independent of $(n,x,y)$, that satisfies the following statements.
Let $x\in D^{(n)}_{(2,3)}$, $y\in H_x$, and $\gamma^{(x)}_{\ell,y}\in MSC_{x,y}$ satisfying $|x\#|\leq \ell \leq |x\#x^R\#|$. Moreover, let  $x=rz$, $x^R=z^Rsr'$, $\ell=|x\#z^Rs|$,  $\gamma^{(x)}_{|r|,y} = \tau v\bot $, and $\gamma^{(x)}_{\ell,y}=\sigma v\bot $ for an appropriate tuple $(r,r',z,s,\sigma,\tau,v)$. If $\ell\neq |x\#|$ and $z\#z^Rs\in TF_{M}(\tau,\sigma)$, then $|\gamma^{(x)}_{y}|\leq d$ holds. Moreover, when
$n$ is sufficiently large, $\ell\neq|x\#|$ also holds.
\end{lemma}

Lemma \ref{TF-bound} roughly states that, if there is an interval between $1$ and $|x\# x^R\#|$ crossing over the first $\#$ in which $M$ transforms a stack symbol $\tau$ to another $\sigma$, the size of stack content is small at the $|x\#x^R\#|$-th cell position.
Figure \ref{fig:stack-history} illustrates a stack history stated in the lemma.

Using Lemma \ref{TF-bound}, we can prove Proposition \ref{size-stack-content} in the following manner.
Since $MSC_{x,y}$ is nonempty, we can take an element  $\gamma^{(x)}_{\ell,y}$ from $MSC_{x,y}$ satisfying  $|x\#|\leq \ell\leq |x\#x^R\#|$. By the size-minimality of $\gamma^{(x)}_{\ell,y}$,
it follows that $|\gamma^{(x)}_{i,y}|\geq |\gamma^{(x)}_{\ell,y}|$ for any cell position  $i$ with $|x\#|\leq i\leq \ell$.
Since $M$ pushes at most one extra stack symbol into the stack, there must be a cell position $i$ satisfying both $1\leq i\leq|x\#|$ and $|\gamma^{(x)}_{i,y}|=|\gamma^{(x)}_{\ell,y}|$. This indicates that, by choosing an appropriate tuple  $(r,r',z,s,\sigma,\tau,u,v)$, we can decompose $x\#x^R\#y$ into
\begin{quote}
\begin{itemize}
\item[(*)] $x=rz$, $x^R=z^Rsr'$, $\ell=|x\#z^Rs|$, $\gamma^{(x)}_{|r|,y} =\tau v\bot $, $\gamma^{(x)}_{\ell,y}= \sigma v\bot $, and
$z\#z^Rs\in TF_{M}(\tau,\sigma)$.
\end{itemize}
\end{quote}
The second part of Lemma \ref{TF-bound} implies that, except for a finite number of $x$'s,  $\ell\neq |x\#|$ always holds. We define $d_2$ to be the total number of those exceptional $x$'s. For the other $x$'s, since $\ell\neq|x\#|$, the first part of  Lemma \ref{TF-bound} then provides the desired constant $d_1$ that upper-bounds  $|\gamma^{(x)}_{y}|$. We therefore obtain the proposition.

\ms

Now, it is time to verify Lemma \ref{TF-bound}. This lemma requires two additional lemmas,  Lemmas \ref{upper-size-s} and \ref{size-bound-of-r}.
In the first lemma given below, we want to show that the size of $s$ in (*) is bounded from above by a certain absolute constant.

\begin{lemma}\label{upper-size-s}
There exists a constant $d_1>0$, independent of $(n,x,y)$,
satisfying the following statements.
Let $x\in D^{(n)}_{(2,3)}$, $y\in H_x$, and $\gamma^{(x)}_{\ell,y}\in MSC_{x,y}$ with $|x\#|\leq \ell \leq |x\#x^R\#|$. Moreover, let $x=rz$, $x^R=z^Rsr'$, $\ell=|x\#z^Rs|$, $\gamma^{(x)}_{|r|,y} = \tau v\bot $, and $\gamma^{(x)}_{\ell,y}=\sigma v\bot $. If $\ell\neq |x\#|$ and $z\#z^Rs\in TF_{M}(\tau,\sigma)$, then $|s|\leq d_1$ holds.
\end{lemma}

\begin{yproof}
Let $(r,r',z,s,\sigma,\tau,u,v)$ satisfy that $x=rz$, $x^R=z^Rsr'$, $\ell=|x\#z^Rs|$, $\gamma^{(x)}_{|r|,y} = \tau v\bot $, and $\gamma^{(x)}_{\ell,y}=\sigma v\bot $. Moreover, we assume that $\ell\neq|x\#|$ and $z\#z^Rs$ belongs to $TF_{M}(\tau,\sigma)$.
From the inequality $\ell\neq|x\#|$, it follows that $z^Rs\neq\lambda$.
Let us assume further that $\gamma^{(x)}_{\ell,y}\in MSC_{x,y}$.
We first claim that the string $s$ can be uniquely determined from the pair $(\tau,\sigma)$.

\begin{claim}\label{segment-stack-content}
Let $z_1\in\{0,1\}^+$ and $s_1\in\{0,1\}^*$ be arbitrary strings. If $z_1\#z_1^Rs_1\in TF_{M}(\tau,\sigma)$, then $s=s_1$ holds.
\end{claim}

Claim \ref{segment-stack-content} uniquely associates  $s$ with $(\tau,\sigma)$, and thus we can define a map from $(\tau,\sigma)$ to $s$. Hence, the  number of all possible strings $s$ is at most $|\Gamma'_{(1,2)}|^2$, which is obviously a constant.
From this fact, we can draw a conclusion that $|s|$ is upper-bounded by an appropriately chosen constant, independent of $(n,x,y)$.

Finally, let us prove Claim \ref{segment-stack-content}.  Toward a contradiction, we assume that $z_1\#z_1^Rs_1\in TF_{M}(\tau,\sigma)$ and $s\neq s_1$.
Let $p_1$ denote any accepting $(1,2)$-computation path generated by $M$ while reading off $rz\#z^Rsr'\# y$. Consider its computation subpath, say, $p_2$ associated with the substring $z\#z^Rs$. By our assumption of $z_1\#z_1^Rs_1\in TF_{M}(\tau,\sigma)$, there exists a computation subpath, say, $p_3$ corresponding to $z_1\#z_1^Rs_1$. Now, along the computation path $p_1$, we replace the subpath $p_2$  by $p_3$.
This produces a new accepting $(1,2)$-computation path on the input $rz_1\#z_1^Rs_1r'\# y$. Thus, we conclude that $(rz_1)^R = z_1^Rr^R = z_1^Rsr' \neq z_1^Rs_1r'$ because of $s\neq s_1$. This means that there is no accepting $(1,2)$-computation path on $rz_1\#z_1^Rs_1r'\# y$, a contradiction. Therefore, Claim \ref{segment-stack-content}  is true.
\end{yproof}

%%%%%

In the second lemma below, we want to show that the size of $r'$ in (*) is also upper-bounded by a certain absolute constant.

\begin{lemma}\label{size-bound-of-r}
There exists a constant $d_2>0$, independent of $(n,x,y)$,
that satisfies the following statements.
Let $x\in D^{(n)}_{(2,3)}$, $y\in H_x$, and $\gamma^{(x)}_{\ell,y}\in MSC_{x,y}$. Moreover, let $x=rz$, $x^R=z^Rsr'$, $y=r''z'$,  $\ell=|x\#z^Rs|$, $\ell'=|x\#x^R\#r''|$,  $\gamma^{(x)}_{|r|,y} = \tau v\bot $,  $\gamma^{(x)}_{\ell,y}=\sigma v\bot $, and $\gamma^{(x)}_{\ell',y}=v\bot $. If $r'\#r''\in TF_{M}(\sigma,\lambda)$, then $|r'|\leq d_2$ holds.
\end{lemma}

\begin{yproof}
Take parameters $(r,r',r'',z,z',u,\tau,\sigma,\ell,\ell')$ as specified in the premise of the lemma and assume that $r'\#r''\in TF_{M}(\sigma,\lambda)$.
Similarly to Claim \ref{segment-stack-content}, we claim that $\sigma$ uniquely determines $r'$.

\begin{claim}\label{unique-r'}
Let $r'_1,r''_1\in\{0,1\}^*$. If $r'_1\#r''_1\in TF_{M}(\sigma,\lambda)$, then $r'_1=r'$.
\end{claim}

Claim \ref{unique-r'} helps us define a map from $\sigma$ to $r'$ since $r'$ is uniquely determined by $\sigma$. This mapping implies that the number of all possible $r'$ is at most $|\Gamma'_{(1,2)}|$.
Since there are at most $|\Gamma'_{(1,2)}|$ such strings $r'$, $|r'|$ must be  bounded from above by a certain constant, independent of $(n,x,y)$.

Claim \ref{unique-r'} itself can be proven by way of contradiction. First, we assume that $r'_1\neq r'$. Since $y\in H_x$, we are focused on the  accepting $(1,2)$-computation path $p_{x,y}$ on the input $x\#x^R\#y$, which equals $rz\#z^Rsr'\#r''z'$.
Since $r'\#r''\in TF_{M}(\sigma,\lambda)$ and $r'_1\#r''_2\in TF_{M}(\sigma,\lambda)$, we can replace its subpath associated with $r'\#r''$ by a subpath generated by $M$ while reading off $r'_1\#r''_1$.
We then obtain another accepting $(1,2)$-computation path on the input $x\#z^Rsr'_1\#r''_1z'$. By the definition of $h_3$ and the choice of $M$, $x^R=z^Rsr'_1$ must hold. On the contrary, we obtain $x^R = z^Rsr' \neq z^Rsr'_1$ from $r'_1\neq r'$. This is a contradiction.
\end{yproof}

%%%%%

Finally, we are ready to prove Lemma \ref{TF-bound} with the help of Lemmas \ref{upper-size-s} and \ref{size-bound-of-r}.

\begin{proofof}{Lemma \ref{TF-bound}}
Let $(r,r',z,s,\sigma,\tau,v,y)$ be given as in the premise of the lemma. Notice that $x=rz$, $x^R=z^Rsr'$, and $\gamma^{(x)}_{\ell,y}=\sigma v\bot $ with $\ell=|x\#z^Rs|$.
Since $\gamma^{(x)}_{\ell,y}\in MSC_{x,y}$, there exists a nonempty string $u$ satisfying $\gamma^{(x)}_{y}= uv\bot$.
Assume that $M$ transforms $\sigma$ to $u$ while reading $r'$.

We first claim that $\ell\neq |x\#|$ for any sufficiently large $n$. Assume otherwise; namely, $\ell=|x\#|$. This assumption yields $z=s=\lambda$, which implies that $x=rz=r$ and $x^R = z^Rsr' = r'$.
In this case, for a certain string $r''$, it follows that $x^R\# r''\in TF_{M}(\sigma,\tau)$ and $\gamma^{(x)}_{\ell',y}=v\bot$, where $\ell'=|x\#x^R\#r''|$.
By Lemma \ref{size-bound-of-r}, we obtain a constant size-upper bound of $r'$; in other words, there exists a constant $d_2$, independent of $(n,x,y)$, satisfying $|r'|\leq d_2$. Since $n=|x|$ is sufficiently large,  $r'$ must be large in size as well. This is a contradiction. Therefore, $\ell\neq|x\#|$ holds.

Hereafter, we assume that $\ell\neq |x\#|$ and $z\#z^Rs\in TF_{M}(\tau,\sigma)$.
Lemma \ref{upper-size-s} further ensures the existence of an appropriate constant $d_1$ for which  $|s|\leq d_1$. Lemma \ref{size-bound-of-r} also shows that $|r'|$ is upper-bounded by a certain constant, say, $d_2$. Since $|r|=|sr'|=|s|+|r'|$ by definition, $|r|$ is bounded from above by $d_1+d_2$.
Let $\sigma_0$ be the stack symbol pushed into the stack at the first step of $M$.
Since $M$ transforms $\sigma_0$ to $\tau v$ while reading $r$ for a certain stack symbol  $\tau$ and the stack increases by at most one, it follows that $|v|$ is upper-bounded by a certain absolute constant. Therefore, since $|r'|\leq d_2$, $|uv\bot|$ is bounded as well.
\end{proofof}

%%%
%%%

We have completed the proof of Proposition \ref{size-stack-content}. As a preparation for a further discussion on Case 1, we provide a useful lemma, which will be used in Section \ref{sec:stack-size}.

\begin{lemma}\label{x2-vs-y-difference}
Let  $x_1,x_2,y\in\{0,1\}^n$ with $x_2\in H_{x_1}$ and $y\in H_{x_2}$. If $x_2\in D_{(2,3)}$ and $x_1\neq x_2$, then
there is no cell position $i$ for which $|x_1|\leq i\leq |x_1\#x_1^R\#|$ and $\gamma^{(x_1)}_{i,x_2} = \gamma^{(x_2)}_{i,y}$.
\end{lemma}

\begin{yproof}
Assume that $x_2\in D_{(2,3)}$ and $x_1\neq x_2$. To lead to a contradiction, we further assume that a position $i$ in the lemma actually exists. Let $j=|x_1\#x_1^R\#|-i$. Now, let us consider two accepting $(1,2)$-computation paths $p_{x_1,x_2}$ and $p_{x_2,y}$.  Since $\gamma^{(x_1)}_{i,x_2}=\gamma^{(x_2)}_{i,y}$, it is possible to swap between subpaths of $p_{x_1,x_2}$ and $p_{x_2,y}$ generated by $M$ while reading substrings $x_1\#(x_1^R)_{j}$ and $x_2\#(x_2^R)_j$, respectively. We then obtain another accepting $(1,2)$-computation path, say, $p$ on the input $x_2\#(x_2^R)_j(x_1^R)_{n-j}\#x_2$. Here, we handle two possible cases.

(Case i) If $(x_2^R)_j(x_1^R)_{n-j} \neq x_2^R$, then the computation path $p$ cannot be an accepting  $(1,2)$-computation path, a contradiction.

(Case ii) If $(x_2^R)_j(x_1^R)_{n-j} = x_2^R$,  then $x_2\#(x_2^R)_j(x_1^R)_{n-j}\#x_2$ equals $x_2\#x_2^R\#x_2$. The obtained computation path $p$ is indeed an accepting $(1,2)$-computation path on $x_2\#x_2^R\#x_2$, and thus $x_2$ must be in $D_{(1,2)}$. This obviously contradicts the choice of $x_2\in D_{(2,3)}$, a contradiction.
\end{yproof}

%%%%%
%%%%%
\subsection{Size of Stack Contents}\label{sec:stack-size}

We continue our discussion on Case 1. In Proposition \ref{size-stack-content}, we have shown that all  but a constant number of  strings $x$ in $D_{(2,3)}$ satisfy the inequality  $|\gamma^{(x)}_{y}|<d_1$ for all strings $y$ in $H_x$.
Toward an intended contradiction, we will further show that there are a large portion of $x$'s in $D_{(2,3)}$ whose corresponding stack contents $\gamma^{(x)}_{y}$ for appropriately chosen strings $y$ are large in size. Together with Proposition \ref{size-stack-content}, we can derive the desired contradiction.

In the subsequent argument, the notation $E_{x}$ expresses the collection of all stack contents $\gamma^{(x)}_{y}$ at the $|x\#x^R\#|$-th cell position (obtained just after reading off $x\#x^R\#$) along an accepting $(1,2)$-computation path $p_{x,y}$ of $M$ on input $x\#x^R\#y$ for each string $y\in H_x$.  Since $\pi$ is fixed, it follows that $1\leq |E_x|\leq |H_x| = 2^{|x|}-2$ because there are at most $|H_x|$ subpaths generated by $M$ while reading inputs $x\#x^R\#y$ with $y\in H_x$.

Prior to a discussion on a general case of $|E_x|\geq1$, we wish to consider a special case where $|E_{x}|$ equals $1$ for any string $x\in D_{(2,3)}$, because this case  exemplifies an essence of our proof for the general case.

%%%
\vs{-2}
\paragraph{\bf I) Special Case of $|E_x|=1$.}

Since $|E_x|=1$, the choice of $y\in H_x$ becomes irrelevant. It is thus possible to drop subscript ``$y$'' altogether and abbreviate, e.g.,   $\gamma^{(x)}_{i,y}$, $\gamma^{(x)}_{y}$, and $u_{x,y}$  as  $\gamma^{(x)}_{i}$, $\gamma^{(x)}$, and $u_{x}$, respectively.
To lead to the desired contradiction, we want to show in Proposition \ref{stack-height-bound} that a large number of strings $x$ in $D_{(2,3)}$ produce stack contents $\gamma^{(x)}_{y}$ of extremely large size for certain strings $y\in H_x$. Now, recall the notation $\Gamma'_{(1,2)}$ that stands for the set $\Gamma_{(1,2)}\cup\{\bot \}$.

\begin{proposition}\label{stack-height-bound}
Given any number $\epsilon\geq0$, it follows that $|\{x\in D_{(2,3)}\mid \exists y\in H_x\,[|\gamma^{(x)}_{y}|\geq (n-2-\epsilon)/\log|\Gamma'_{(1,2)}|]\}|\geq |D_{(2,3)}|(1-2^{-\epsilon})$.
\end{proposition}

To prove Proposition \ref{stack-height-bound},  let us consider two stack contents $\gamma^{(x_1)}_{x_2}$ and $\gamma^{(x_2)}_{x_1}$ associated with two distinct strings $x_1,x_2\in D_{(2,3)}$. By choosing $i=|x_1\#x_1^R\#|$ in Lemma  \ref{x2-vs-y-difference}, we immediately obtain  $\gamma^{(x_1)}_{x_2} \neq \gamma^{(x_2)}_{x_1}$. We then reach the following conclusion.

\begin{lemma}\label{distinct-pair}
For every distinct pair $x_1$ and $x_2$ in $D_{(2,3)}$, it follows that $\gamma^{(x_1)}\neq \gamma^{(x_2)}$.
\end{lemma}

 To simplify our notation further,  we write
$A_d$ to express the set $\{x\in D_{(2,3)} \mid \exists y\in H_x\,[|\gamma^{(x)}_{y}|\geq d]\}$ for each chosen number $d\in\nat^{+}$. With this notation, Proposition \ref{stack-height-bound} is equivalent to the assertion that $|A_{(n-2-\varepsilon)/\log|\Gamma'_{(1,2)}|}|\geq |D_{(2,3)}|(1-2^{-\varepsilon})$. Associated with $A_d$, we define $B_d=\{x\in D_{(2,3)} \mid \forall y\in H_x\,[|\gamma^{(x)}_{y}|<d]\}$.
Note that $D_{(2,3)} = A_d\cup B_d$ and $A_d\cap B_d = \setempty$ for any number $d\in\nat^{+}$. In the following lemma, we present a lower bound on the cardinality of $A_d$.

\begin{lemma}\label{bound-of-A_d}
For any constant $d\in\nat^{+}$, it follows that $|A_d|\geq |D_{(2,3)}| - 2|\Gamma'_{(1,2)}|^d$.
\end{lemma}

\begin{yproof}
Since  $\{A_d,B_d\}$ partitions $D_{(2,3)}$, it follows that $|A_d|= |D_{(2,3)}| - |B_d|$. To prove the lemma, let us concentrate on $B_d$. From $|E_x|=1$, $B_d$ coincides with $\{x\in D_{(2,3)}\mid |\gamma^{(x)}|<d\}$.  Notice that each $\gamma^{(x)}$ belongs to $(\Gamma'_{(1,2)})^m$ for a  certain number $m$ with $m\leq d-1$. Consider a mapping $h$ from $x$ to $\gamma^{(x)}$. Induced from $B_d$, we define  $\bar{B}_d =\{x\in B_d\mid x^R=x\}$. The function $h$ is 1-to-1 on $\bar{B}_d$ and, by Lemma \ref{distinct-pair}, it is  also 1-to-1 on at least a half of elements in $B_d-\bar{B}_d$. Hence,  it follows that $|B_d|/2 \leq \sum_{j=0}^{d-1}|\Gamma'_{(1,2)}|^j=|\Gamma'_{(1,2)}|^d$. We conclude that  $|A_d|= |D_{(2,3)}| - |B_d| \geq |D_{(2,3)}| -2|\Gamma'_{(1,2)}|^d$, as requested.
\end{yproof}

With the help of Lemma \ref{bound-of-A_d}, Proposition \ref{stack-height-bound} is now easy to prove.

\begin{proofof}{Proposition \ref{stack-height-bound}}
For simplicity, we use $d$ to denote  $(n-2-\epsilon)/\log|\Gamma'_{(1,2)}|$, which equals $\log_{|\Gamma'_{(1,2)}|}2^{n-2-\epsilon}$. Our goal is to show that $|A_d|\geq |D_{(2,3)}|(1-2^{-\epsilon})$.
By Lemma \ref{bound-of-A_d}, we obtain $|A_d|\geq |D_{(2,3)}|- 2|\Gamma'_{(1,2)}|^d$. Notice that, by the definition, $|\Gamma'_{(1,2)}|^d = 2^{n-2-\varepsilon} \leq |D_{(2,3)}|\cdot 2^{-1-\varepsilon}$, where the last inequality comes from our assumption of $|D_{(2,3)}|\geq 2^n/2$. As a result, we obtain $|A_d| \geq |D_{(2,3)}|- 2|\Gamma'_{(1,2)}|^d \geq |D_{(2,3)}|(1-2^{-\epsilon})$, as requested.
\end{proofof}

To finish this special case, let $x=rz$, $x^R=z^Rsr'$,  $\gamma^{(x)} = uv\bot $, and $\gamma^{(x)}_{\ell}=\sigma v\bot $ with $\ell=|x\#z^Rs|$. Assume that $M$ transforms $\sigma$ to $u$ while reading $r'$.
Proposition \ref{size-stack-content} shows that, for most of $x$'s,
$|uv\bot |$ is upper-bounded by a certain constant, independent of $(n,x,y)$.
However, by setting, \eg $\epsilon=98$, Proposition \ref{stack-height-bound} yields  $|uv\bot |\geq (n-100)/\log|\Gamma'_{(1,2)}|$ for at least the $2/3$-fraction of $x$'s in $D_{(2,3)}$. Since $n$ is sufficiently large, $|uv\bot |$ cannot be bounded from above by any absolute constant. Therefore, we obtain a clear contradiction.

%%%%%
%%%%%
\vs{-1}
\paragraph{\bf II) General Case of $|E_x|\geq1$.}

We have already shown how to cope with the case of $|E_x|=1$ for all strings $x\in D_{(2,3)}$. Hereafter, we will discuss a general
case where
$|E_{x}|\geq 1$ holds for any $x\in D_{(2,3)}$.
Our goal is to prove the correctness of the following statement.

\begin{proposition}\label{symbol-length-small}
Let $d>0$. All but $2(|\Gamma'_{(1,2)}|^d+1)$ strings  $x$ in $D_{(2,3)}$ satisfy the following: there exists a stack content $\tau \in E_x$ for which $\tau$ contains at least $d$ symbols; namely, $|\{x\in D_{(2,3)}\mid \exists\,\tau\in E_x\,[|\tau|\geq d]\}|\geq |D_{(2,3)}| -2|\Gamma'_{(1,2)}|^d-2$.
\end{proposition}

We first give basic notions and notation needed for the proof of Proposition \ref{symbol-length-small}. With a fixed number $n$, let
$G_n$ denote a specific undirected graph $(V_n,E_n)$ in which $V_n= \{(i,j)\mid i,j\in[n],i\neq j\}$ and, for any two vertices $v_1=(i_1,j_1)$ and $v_2=(i_2,j_2)$ in $V_n$, $(v_1,v_2)$ belongs to $E_n$ if either $i_1=j_2$ or $i_2=j_1$. A \emph{coloring} of $G_n$ is a function $\phi:V_n\to C$ for a certain finite set $C$.
Given a number $t\in\nat^{+}$, $G_n$ is said to be \emph{$t$-colorable} if there exists a coloring $\phi:V_n\to C$ with $|C|=t$ for which no single color is assigned to two adjacent vertices (i.e., no edge $(v_1,v_2)\in E$ satisfies $\phi(v_1)=\phi(v_2)$).
The \emph{chromatic number} of $G_n$, denoted by $\chi(G_n)$, is the smallest number $t\in\nat^+$ that makes $G_n$ be  $t$-colorable.

\begin{lemma}\label{chromatic}
For any $n\in\nat^{+}$ with $n\geq3$, $\chi(G_n)$ equals $n$.
\end{lemma}

\begin{proof}
Firstly, we assert that $\chi(G_n)\leq n$. To achieve this goal, we define a special coloring $\phi$ as follows. Let $C$ be the color set $\{c_1,c_2,\ldots,c_n\}$. For any pair $i,j\in[n]$ with $i\neq j$, we set $\phi((i,j))=c_j$. It then follows that, for any two distinct vertices $v_1=(i_1,j_1)$ and $v_2=(i_2,j_2)$ in $G_n$, if $\phi(v_1)=\phi(v_2)$, then $j_1=j_2$ holds; thus, we obtain $(v_1,v_2)\notin E_n$ because, otherwise, either $i_1=j_1$ or $i_2=j_2$ follows.
Therefore, $\phi$ is a valid coloring of $G_n$. Since $|C|=n$, we obtain $\chi(G_n)\leq n$.

To show that $\chi(G_n)\geq n$, on the contrary, we start with calculating the ``independent number'' of $G_n$. An \emph{independent set} $A$ of $G_n$ is a set of vertices of $G_n$ such that any two distinct vertices in $A$ cannot be adjacent in $G_n$. The \emph{independent number} $\alpha(G_n)$ is the maximum size of any independent set of $G_n$. It is immediate that $\chi(G_n) \alpha(G_n) \geq |V|$. We estimate $\alpha(G_n)$ in the following claim.

\begin{claim}\label{independent}
For each $n\in\nat^{+}$ with $n\geq3$, $\alpha(G_n)=n-1$.
\end{claim}

From this claim, since $|V_n|=n(n-1)$ and $\alpha(G_n)=n-1$, we conclude that
$\chi(G_n)\geq \frac{|V_n|}{\alpha(G_n)} = \frac{n(n-1)}{n-1} = n$, as requested.

To prove Claim \ref{independent}, we first show that $\alpha(G_n)\geq n-1$. For this purpose, let us consider the set $A=\{(1,i)\in V_n \mid i\in[2,n]_{\integer}\}$. Clearly, $A$ is an independent set of $G_n$. Since $|A|=n-1$, we immediately obtain $\alpha(G_n)\geq n-1$.
Next, we intend to verify that $\alpha(G_n)\leq n-1$. Toward a contradiction, we assume that $\alpha(G_n)\geq n$.
Take an independent set, say, $B$ of $G_n$ of cardinality $n$. We assume that $B$ has the form $\{(k_i,j_i)\mid i\in[n]\}$. Since $|B|=n$, it is not difficult to show that there are two numbers $i_1,i_2\in[n]$ for which either  $k_{i_1} = j_{i_2}$ or $k_{i_2}=j_{i_1}$ holds. This is a contradiction. Thus, we conclude that $\alpha(G_n)\leq n-1$. This completes the proof of the claim.
\end{proof}

Let us return to the proof of Proposition \ref{symbol-length-small}.

\begin{proofof}{Proposition \ref{symbol-length-small}}
As done in I), we set $A_d = \{x\in D_{(2,3)}\mid \exists y\in H_x\, [|\gamma^{(x)}_{y}|\geq d]\}$ and define $B_d=\{x\in D_{(2,3)}\mid \forall y\in H_x\,[|\gamma^{(x)}_{y}|<d]\}$ so that $D_{(2,3)}=A_d\cup B_d$.
With these notations, the proposition asserts that $|A_d|\geq |D_{(2,3)}|-2(|\Gamma'_{(1,2)}|^d+1)$, or equivalently, $|B_d|\leq 2(|\Gamma'_{(1,2)}|^d+1)$.
We further restrict $B_d$ as $\hat{B}_d = \{x\in B_d\mid x<x^R\}$, where ``$<$'' is the lexicographic ordering on $\{0,1\}^n$. Obviously,
$|B_d|\leq 2|\hat{B}_d|$ follows. To obtain the proposition,
it therefore suffices to prove the statement: (*) $|\hat{B}_d|\leq |\Gamma'_{(1,2)}|^d+1$. In what follows, we wish to prove this statement (*).

Let $m=|\hat{B}_d|$ for simplicity. We express all elements of $\hat{B}_d$ as $\{x_1,x_2,\ldots,x_m\}$ and then identify it with the integer set $\{1,2,\ldots,m\}$ ($=[m]$).
We then introduce an undirected graph $G=(V,E)$ as follows. Let $V=\{(x,y)\mid x,y\in[m],x\neq y\}$ and define $E$ to be composed of all edges $((x_1,x_2),(x_2,y_2))$ such that (i) $(x_1,y_1),(x_2,x_2)\in V$ and (ii) either $x_1=y_2$ or $x_2=y_1$. We set
$C=\{\gamma^{(x)}_{y}\mid (x,y)\in V\}$ and define a function $\phi: V\to C$ by setting $\phi(x,y)=\gamma^{(x)}_{y}$, where $x$ and $y$ are seen as associated elements in $\hat{B}_d$. We assert the following claim concerning $\phi$.

\begin{claim}\label{coloring}
For any $x,y,z$, if $(x,y)$ and $(y,z)$ are vertices in $V$, then $\phi(x,y)\neq \phi(y,z)$ holds.
\end{claim}

To prove the claim, we assume otherwise. Take three elements $x,y,z$ satisfying that $(x,y),(y,z)\in V$ and $\phi(x,y)=\phi(y,z)$.
Note that $x,y\in D_{(2,3)}$, $y\in H_x$, and $z\in H_y$. By Lemma \ref{x2-vs-y-difference} with $i=|x\# x^R \#|$, we conclude that $\gamma^{(x)}_{i,y}\neq \gamma^{(y)}_{i,z}$. From this inequality, we obtain $\phi(x,y)\neq \phi(y,z)$. However, this is obviously a contradiction. Therefore, the claim is true.

By Claim \ref{coloring}, $\phi$ turns out to be a coloring of $G$; thus, we obtain  $\chi(G)\leq |C|$. Since $C\subseteq (\Gamma'_{(1,2)})^{<d}$, it follows that $\chi(G) \leq |C|\leq \sum_{i=0}^{d-1}|\Gamma'_{(1,2)}|^i = |\Gamma'_{(1,2)}|^d$.
Furthermore, Lemma \ref{chromatic} implies that the chromatic number of $G$ is exactly  $m-1$, which equals $|\hat{B}_d|-1$. Thus, we conclude that $|\hat{B}_d|\leq |C|+1\leq |\Gamma'_{(1,2)}|^d+1$, as requested.
\end{proofof}

Finally, let us close Case 1 by drawing a contradiction. Here, we set $d= \floors{\log_{|\Gamma|}(|D_{(2,3)}|/8)}$.
A simple calculation yields
$2(|\Gamma'_{(1,2)}|^d+1) \leq 4|\Gamma'_{(1,2)}|^d \leq |D_{(2,3)}|/2$. Moreover, since $|D_{(2,3)}|\geq 2^n/2$, it follows that $d\geq \log_{|\Gamma'_{(1,2)}|}\frac{|D_{(2,3)}|}{8}-1 \geq \frac{n-4}{\log{|\Gamma'_{(1,2)}|}}-1$.  By Proposition \ref{symbol-length-small}, for at least $|D_{(2,3)}|-2|\Gamma'_{(1,2)}|^d-2$ ($\geq |D_{(2,3)}|/2$) elements $x$ in $D_{(2,3)}$,  an appropriately chosen string $y\in H_x$ makes  $\gamma^{(x)}_{y}$ satisfy $|\gamma^{(x)}_{y}|\geq d$, which further implies  $|\gamma^{(x)}_{y}| \geq (n-4)/\log{|\Gamma'_{(1,2)}|}-1$.
Proposition \ref{size-stack-content}, however,  indicates that, except for at most $d_2$ elements in $D_{(2,3)}$, all strings $x\in D_{(2,3)}$ satisfy $|\gamma^{(x)}_{y}|\leq d_1$ for any choice of $y$ in $H_x$, where $d_1$ and $d_2$ are absolute constants, not depending on $(n,x,y)$.
Since there are infinitely many $n$ satisfying $|D_{(2,3)}|\geq 2^n/2$, for any sufficiently large $n$, there exists a string $x\in D_{(2,3)}$ for which $|\gamma^{(x)}_{y}|\geq (n-4)/\log|\Gamma'_{(1,2)}|-1$ and $|\gamma^{(x)}_{y}|\leq d_1$.
This leads to a clear contradiction, as requested; therefore, this closes Case 1.

%%%%%%%
%%%%%%
\subsection{Case 2: $D^{(n)}_{(1,2)}$ is Large for Infinitely Many Lengths $n$}\label{sec:case-two}

We have already handled Case 1 in Sections \ref{sec:case-one}--\ref{sec:stack-size}. To complete the proof of Proposition \ref{h3-no-refinement}, nevertheless, we still need to deal with the remaining second case where $\{n\in\nat^{+}\mid |D^{(n)}_{(2,3)}|\geq 2^n/2\}$ is a finite set, implying that  $|D^{(n)}_{(1,2)}|> 2^{n}/2$ holds for all but finitely many $n\in\nat$. Instead of managing an argument similar to Case 1, we instead make a quite different approach. Let us recall from Section \ref{sec:colored-automata}  the introduction of the colored automaton $M = (Q,\Sigma,\{\cent,\dollar\},\Gamma,I_3,\delta,q_0,\bot, Q_{acc},Q_{rej})$ in an almost ideal shape that computes $g$. We set
$Q_{acc}=\{q_{acc}\}$ and $Q_{rej}=\{q_{rej}\}$.
Before starting the intended proof for the second case, we present a general statement, ensuring the existence of another colored automaton that can ``simulate'' $M$ on inputs \emph{in a backward fashion}.

\begin{proposition}\label{reversing-machine}
There exists a colored automaton $M^R$ that satisfies the following for any three strings $x_1,x_2,x_3\in\Sigma^*$ and for any $(i,j)\in I_3$:
$M$ accepts $x_1\#x_2\#x_3$ along an accepting $(i,j)$-computation path if and only if $M^R$ accepts $x_3^R\#x_2^R\#x_1^R$ along an accepting $(4-j,4-i)$-computation path.
\end{proposition}

\begin{yproof}
Under our assumption that the colored automaton $M$ in an almost ideal shape computes $g$, we wish to describe the desired colored automaton $M^R  = (Q^R,\Sigma,\{\cent',\dollar'\},\Gamma,I_3, \delta^R,q'_0,\{q'_{acc}\},\{q'_{rej}\})$, where, for clarity reason, we use two different special symbols $\cent'$ and $\dollar'$ to stand for the  endmarkers of $M^R$. Let $w$ denote any input
of the form $x_1\#x_2\#x_3$ given to $M$.
Since $M$ is in an almost ideal shape, $M$ must empty its stack before or at scanning $\dollar$ along any computation path.
To make our proof simpler, we further modify $M$ so that $M$ never enters any halting state before reading $\dollar$.

Intuitively, the desired machine $M^R$ works as follows. We start with the unique accepting state $Q_{acc}$ of $M$ by placing a tape head onto the endmarker $\dollar$, and nondeterministically traverse a computation of $M$ on $w$ backward by moving its tape head leftward from $\dollar$ to $\cent$. To maintain the color scheme, we initially guess a color and use only stack symbols of the same color during the reverse simulation of $M$. If we successfully enter the initial state $q_0$ of $M$ after reaching $\cent$, then we accept the input; otherwise, we reject the input at scanning $\cent$.

More formally, the colored automaton $M^R$ takes the input of the form $w^R$ ($ = x_3^R\#x_2^R\#x_1^R$). We set $q'_0=q_{acc}$ and $q'_{acc}=q_0$. The machine $M^R$ starts with this initial state $q'_0$ with its tape head in the $0$th cell with $\cent'$. The machine $M^R$ guesses (\ie chooses nondeterministically) a color $(i,j)\in I_3$ and remembers it until the end of computation. Let us define the transition function $\delta^R$ of $M^R$.
Assume that at present $M^R$ is in inner state $q$ with stack content $\gamma=\xi z$ and its input tape head scanning a cell containing $\sigma$. It is important to note that the color $(i,j)$ for $M$ is translated to  color $(4-j,4-i)$ for $M^R$ due to the use of the reversal $w^R$ of the original input $w$.
Now, $M$ is assumed to make a transition of the form $(q,\eta)\in\delta(p,\sigma,\xi)$ with $\eta\in \Gamma^{\leq 2}$.
We discuss three possible cases separately, depending
on the size of $\eta$.

\renewcommand{\labelitemi}{$\circ$}
\begin{enumerate}\vs{-1}
  \setlength{\topsep}{-2mm}%
  \setlength{\itemsep}{0mm}%
  \setlength{\parskip}{0cm}%

\item[(1)]  When $\eta$ has the form $\eta_1\eta_2$ with  $|\eta|=2$, the stack content of $M$ changes from $\xi z\bot$ to $\eta_1\eta_2 z\bot$ for a certain string $z$. In this case,  $M^R$ removes $\eta$ and changes its inner state from $q$ to $p$ by the following two steps. We first introduce a new inner state $\hat{q}$ associated with $q$ and then define $(\hat{q},\lambda)\in\delta^R(q,\sigma,\eta_1)$ and $(p,\xi)\in\delta^R(\hat{q},\lambda,\eta_2)$. Notice that the second step is a $\lambda$-move.

\item[(2)] When $|\eta|=1$, $M$ changes its stack content from $\xi z\bot$ to $\eta z\bot$. We simply define $(p,\xi)\in\delta^R(q,\sigma,\eta)$.

\item[(3)] When $\eta=\lambda$, the stack content $\xi z\bot$ of $M$ is modified to $z\bot$. We then define $(p,\xi\tau)\in\delta^R(q,\sigma,\tau)$ for any stack symbol $\tau\in \Gamma$ of the same color as $\xi$.
\end{enumerate}
At last, when scanning $\dollar'$, if $q\neq q_0$, then we define  $\delta^R(q,\dollar',\tau) = \{(q'_{rej},\tau)\}$ for any $q\in Q-Q_{halt}$. Otherwise, we define $\delta^R(q_0,\dollar',\tau) = \{(q'_{acc},\tau)\}$.

It is not difficult to verify by the definition that $M^R$ correctly ``simulates'' $M$ in a reversible way.
\end{yproof}

Let us return to our proof for Case 2, in which,  by running $M$ on inputs of the form $x\# x^R\# x$ for $x\in\{0,1\}^n$, we  obtain $|D^{(n)}_{(1,2)}|>2^n/2$ for infinitely many numbers $n\in\nat$.
Proposition \ref{reversing-machine} provides us with another colored automaton $M^R$ that ``simulates'' $M$ in a reversible manner on any input written in reverse. By Lemma \ref{modify-ideal-shape},  we can convert $M^R$ to one in an almost ideal shape. For the ease of notation, we use the same notation $M^R$ to express the converted machine. A counterpart of $D_{(1,2)}$, denoted by $D^R_{(2,3)}$, is obtained by running $M^R$, instead of $M$, on inputs of the form  $x\#x^R\# x$. By the construction of $M^R$ in the proof of  Proposition \ref{reversing-machine}, we can conclude that  $|D^R_{(2,3)}|>2^n/2$ holds for infinitely many numbers $n\in\nat$. Now, we apply  an argument used for
Case 1  to $D^R_{(2,3)}$, and we then drive an intended contradiction. We have therefore completed
the entire proof of Proposition \ref{h3-no-refinement}.

%%%%%%%%%%%%%%%%%%%%%%%%%%%%
%%%%%%%%%%%%%%%%%%%%%%%%%%%%
\section{Future Challenges}\label{sec:future}

Throughout this paper, we have discussed a question of whether multi-valued partial functions can be refined by certain single-valued partial functions. For NFA functions, Kobayashi \cite{Kob69} solved this refinement question affirmatively. Konstantinidis, Santean, and Yu \cite{KSY07} tackled the same question for CFL functions and obtained a partial solution but left the entire question open. In this paper, we have answered this question negatively by proving that $\ucfltwov\not\sqsubseteq_{ref}\cflsv$ (Theorem \ref{UCFL2V-by-CFLSV}).
In a natural, analogous way, we can expand our interest from $\ucfltwov$ and $\cflsv$ to unambiguous $(k(n)+1)$-valued and $k(n)$-valued function families, $\ucflkv{(k(n)+1)}$ and $\ucflkv{k(n)}$ for any appropriately chosen function $k:\nat\to\nat^{+}$, where ``$n$'' refers to input length. Here,  we wish to raise a more general question of the following form regarding unambiguous $(k(n)+1)$-valued CFL functions.

\begin{openproblem}
Is it true that $\ucflkv{(k(n)+1)}\sqsubseteq_{ref}\cflkv{k(n)}$?
\end{openproblem}

It is not clear that the proof argument of this paper can be straightforwardly extended to solve this general question. Nevertheless, we  conjecture that a negative solution is possible for any ``reasonable''  function $k(n)\in 2^{O(n)}$.

As another type of  extension mentioned
in Section \ref{sec:introduction},
Yamakami \cite{Yam14b} partially settled the refinement question for $\sigmacflmv{k}$ in the CFLMV hierarchy when $k\geq3$, where the CFLMV hierarchy was defined in \cite{Yam14b} as follows. Given a function class $\FF$, its complement class $\co\FF$ is composed of all functions $f:\Sigma^*\to\Gamma^*$ such that there exist a function $g\in \FF$, two constants $a,b\geq0$, and a number $n_0\in\nat$ for which $f(x)=\Gamma^{\leq a|x|+b}-g(x)$  for all strings $x$ in $\Sigma^{\geq n_0}$.
Inductively, let $\sigmacflmv{1}=\cflmv$,  $\picflmv{k}=\co\sigmacflmv{k}$, and $\sigmacflmv{k+1}=\cflmv_{T}^{\sigmacfl{k}}$ for $k\geq1$, where $\cflmv_{T}^{\CC}$ is the collection of all multi-valued partial functions that are computed by \emph{oracle npda's}, which are allowed to access oracle $A\in\CC$ adaptively, running in $O(n)$ time for inputs of length $n$.
With these notations, Theorem \ref{CFL2V-refine-CFLSV} can be rephrased as  $\sigmacflmv{1}\not\sqsubseteq_{ref}\sigmacflsv{1}$.
However, as noted in Section \ref{sec:introduction},  we do not know any answer to the following question regarding the 2nd level of the CFLMV hierarchy.

\begin{openproblem}
Does $\sigmacflmv{2}\sqsubseteq_{ref}\sigmacflsv{2}$ hold?
\end{openproblem}

We conjecture that this refinement question could be solved negatively as well.

%%%%%%%%%%%%%%%%%%%%%%%
%%%%%%%%%%%%%%%%%%%%%%%
\section*{Appendix: Proof of Lemma  \ref{modify-ideal-shape}}

Lemma \ref{modify-ideal-shape} guarantees that it suffices for us to consider only colored automata \emph{in an almost ideal shape}. In Section \ref{sec:main-theorem}, we have used this fact extensively; however, we have left the lemma unproven in Section \ref{sec:colored-automata}. In what follows, we provide a sketch of its proof, in which we render a procedure of how to convert any colored automaton $M$ to its ``equivalent'' colored automaton $N$ in an almost ideal shape. A fundamental idea for this procedure comes from the conversion of any context-free grammar to Greibach Normal Form (see, e.g., \cite{HU79}).

Let $M=(Q,\Sigma,\{\cent,\dollar\},\Gamma,C,\delta,q_0,\bot, Q_{acc},Q_{rej})$ denote any colored automaton with a color partition $\{\Gamma_{\xi}\}_{\xi\in C}$ of $\Gamma$ except for $\bot $. In what follows, we will construct another colored automaton $N$ in an almost ideal shape of the form $(Q',\Sigma,\{\cent,\dollar\}, \Gamma',C,\delta',q_0,\bot, \{q'_{acc}\},\{q'_{rej}\})$ that can simulate  $M$.
Hereafter, we will describe how to convert $M$ into $N$ step by step. To clarify each step of modification, with a slight abuse of the symbols, we want to use $\delta$, $Q$, and $\Gamma$ to indicate a function and sets that have been already modified during the previous step and we use $\delta'$ and $\Gamma'$ for their newly modified versions obtained at the current step. Note that the conversion method given below also works for the case where all computation paths are not required to terminate in linear time.

\s

(1) As a basic transformation, we first remove from $\Gamma$ all stack symbols that never be used in any computation of $M$ on an arbitrary input. Those symbols are called \emph{useless}. Next, we restrict $Q_{acc}$ and $Q_{rej}$ to $\{q'_{acc}\}$ and $\{q'_{rej}\}$, respectively, by reassigning all $q\in Q_{acc}$ (resp., $q\in Q_{rej}$) to $q'_{acc}$ (resp., $q'_{rej}$).
Finally, we  modify the machine so that it never enters any halting state before scanning $\dollar$. For this purpose, we postpone the timing of entering any halting state  by introducing a dummy accepting state and a dummy rejecting state  and by staying in those inner states while an input-tape head moves to the right until
it eventually arrives at $\dollar$.

(2) We then convert $Q$ to $Q'=\{q_0,q,q'_{acc},q'_{rej}\}$ by encoding the information on the changes of inner states into stack symbols in a nondeterministic fashion. We translate
(i) a transition of the form $(r,c_1c_2\cdots c_k)\in \delta(q_0,\cent,\bot )$ to a new transition $(q,\track{rp_1}{c_1}\track{p_1p_2}{c_2}\cdots \track{p_kp_{k+1}}{c_k}\track{p_{k+2}p_{k+3}}{\bot })\in\delta'(q_0,\cent,\bot )$ for all possible inner states $p_1,p_2,\ldots,p_{k+3}\in Q$ satisfying $p_2,p_4,\ldots,p_{k+2}\notin Q_{halt}$,
(ii) a transition of the form  $(r,c_1c_2\cdots c_k)\in\delta(p,\sigma,a)$ with $\sigma\in\check{\Sigma}\cup\{\lambda\}$ and $p\neq q_0$ to $(q,\track{rp_1}{c_1}\track{p_1p_2}{c_2}\cdots \track{p_kp_{k+1}}{c_k})\in \delta'(q,\sigma,\track{pr}{a})$, and
(iii) a transition of the form $(q',w)\in\delta(p,\sigma,a)$ with $q'\in Q_{halt}$ and $\sigma\in\check{\Sigma}\cup\{\lambda\}$ to $(q',\track{pq'}{a})\in\delta'(q,\sigma,\track{pq'}{a})$, where $\track{pr}{a}$, $\track{p_ip_{i+1}}{c_i}$, $\track{pq'}{a}$, etc. are all new stack symbols. Here, we paint those new symbols in the same color as $a$ and $c_i$ have.

(3) We supplement all missing transitions (if any) with special transitions that directly guide to the unique rejecting state $q'_{rej}$.

(4) We eliminate all transitions of the form $(q,\lambda)\in\delta(q,\lambda,a)$. After this step,  no $\lambda$-move deletes a stack symbol. This elimination is done by finding  so-called \emph{nullable symbols} as follows.
A stack symbol $a$ is \emph{nullable} if there is a transition of the form $(q,\lambda)\in\delta(q,\lambda,a)$.
Note that, when a transition $(q,b_1b_2\cdots b_k)\in\delta(q,\lambda,a)$ exists and all $b_i$'s are nullable, $a$ is also nullable.
Associated with each transition $(q,c_1c_2\cdots c_k)\in\delta(q,\lambda,a)$, we include all transitions of the form $(q,e_1e_2\cdots e_k)\in\delta'(q,\lambda,a)$ satisfying the following three conditions:
(i)  $e_1e_2\cdots e_k\neq\lambda$, (ii) $e_i=c_i$ if $c_i$ is not nullable,  and (iii) $e_i\in\{c_i,\lambda\}$ if $c_i$ is nullable.

(5) We remove all transitions that make \emph{single-symbol replacement}, namely, transitions of the form  $(q,b)\in\delta(q,\lambda,a)$ for $a,b\in\Gamma$.  From the existing set of transitions, we first choose all transitions that do not have the above form and make them new transitions of $\delta'$. We then define a new transition  $(q,w)\in\delta'(q,\lambda,a)$ if a transition $(q,w)\in\delta(q,\lambda,b)$ exists and $M$ transforms $a$ to $b$ along a certain computation subpath without using the transition  $(q,b)\in\delta(q,\lambda,a)$.

(6) We delay the start of a loop given by a transition of the form $(q,au)\in\delta(q,\lambda,a)$. The following \emph{loop-delay conversion} eliminates this form entirely. Assume that there are transitions $(q,au)\in\delta(q,\lambda,a)$ and $(q,w)\in\delta(q,\sigma,a)$ with $\sigma\in\check{\Sigma}\cup\{\lambda\}$ and $w\notin a\Gamma^*$. We introduce a new symbol $b$ (in the same color as $a$'s) and introduce new transitions $(q,u)\in\delta'(q,\lambda,b)$, $(q,ub)\in\delta'(q,\lambda,b)$, and $(q,wb)\in\delta'(q,\sigma,a)$.

(7) We eliminate all $\lambda$-moves made while reading inputs (including the endmarkers). Let $\Gamma=\{a_0,a_1,\ldots,a_k\}$ be a stack alphabet defined at the previous step with $a_0=\bot $. This step is composed of the following three substeps (i)--(iii).

(i) First, we inductively modify the transitions (and also adding extra new symbols) so that, for any pair $i,j\in[0,k]_{\integer}$, $(q,a_iu)\in \delta'(q,\lambda,a_j)$ implies $i>j$. For each index $j=0,1,\ldots,k$, choose $i=0,1,\ldots,j-1$ sequentially and conduct the following modifications (a)--(b).
(a) When a transition $(q,a_iu)\in\delta(q,\lambda,a_j)$ exists for $u\in\Gamma^*$, we include  transitions $(q,wu)\in\delta'(q,\sigma,a_j)$ and $(q,w)\in\delta'(q,\sigma,a_i)$ for each transition $(q,w)\in\delta(q,\sigma,a_i)$ with $\sigma\in\check{\Sigma}\cup\{\lambda\}$ (and $w\notin a_i\Gamma^*$ by (6)).
(b) Next, for each transition of the form $(q,a_ju)\in\delta(q,\lambda,a_j)$ (possibly) generated in (a),
we apply the loop-delay conversion of (6) by introducing a new symbol $b_j$ whose color is set to be the same as $a_j$.

(ii) For each index $j=k-1,k-2,\ldots,0$ chosen sequentially,  if $(q,a_iu)\in\delta(q,\lambda,a_j)$ with $i>j$ exists, then we add  $(q,wu)\in\delta'(q,\sigma,a_j)$ for each transition  $(q,w)\in\delta(q,\sigma,a_i)$ with $\sigma\in\check{\Sigma}$ (notice that $w$ does not begin with a symbol in $\{a_{i+1},a_{i+2},\ldots,a_k\}$).

(iii) For the newly added $b_j$'s, we have only transitions of the form $(q,w)\in\delta(q,\lambda,b_j)$ with $w$ beginning with $a_i$'s. Associated with each transition of the form  $(q,a_iu)\in\delta(q,\lambda,b_j)$, we include $(q,wu)\in\delta'(q,\sigma,b_j)$ for each transition  $(q,w)\in\delta(q,\sigma,a_i)$ with $\sigma\in\check{\Sigma}$.

(8) Finally, we reduce to at most $2$ the number of stack symbols pushed simultaneously into the stack. Let us consider the set $A=\{w\in \Gamma^* \mid \exists p,p'\in Q\, \exists a\in\Gamma\, \exists \sigma\in\check{\Sigma}\, [(p',w)\in\delta(p,\sigma,a)]\}\cup \Gamma$. Let $\Gamma'$ be composed of all symbols $[u]$ for all prefixes $u$ of $w$ in $A$ and all symbols $[bw]$ for $b\in\Gamma-\{\bot \}$ and $w\in A$.  For simplicity, we set $[\lambda]=\lambda$.
If a transition $(q,w)\in\delta(q_0,\cent,\bot )$ exists, then we include a transition $(q,[w][\bot ])\in \delta'(q_0,\cent,[\bot ])$, where $[\bot ]$ is the new bottom marker. For each transition  $(q,w)\in\delta(q,\sigma,b)$ with $\sigma\in\check{\Sigma}$,  we define $(q,[w][u])\in\delta'(q,\sigma,[bu])$ for any $u$ satisfying   $[bu]\in \Gamma'$.  Moreover, if a transition $(q,\lambda)\in\delta(q,\sigma,b)$ exists, then we include transitions  $(q,[u])\in\delta'(q,\sigma,[bu])$ for any $u$ with $[bu]\in\Gamma'$. The colors of $[w]$, $[u]$, and  $[u]$ are the same as those of all symbols in $w$, $u$, and $bu$, respectively.

%%%%%%%%%%%%%%%%%%%%%%%%%%
%%%%%%%%%%%%%%%%%%%%%%%%%%
%%%%%%%%%%%%%%%%%%%%%%%%%%
\let\oldbibliography\thebibliography
\renewcommand{\thebibliography}[1]{%
  \oldbibliography{#1}%
  \setlength{\itemsep}{0pt}%
}
\bibliographystyle{plain}

%%%%%%%%%%%%%%%%%%%%%%%%%%%%%%%%%%%%%%%%%%%%%%%%%%%
%%%%%%%%%%%%%%%%%%%%%%%%%%%%%%%%%%%%%%%%%%%%%%%%%%%
%%%%%%%%%%%%%%%%%%%%%%%%%%%%%%%%%%%%%%%%%%%%%%%%%%%
%%%%%%%%%%%%%%%%%%%%%%%%%%%%%%%%%%%%%%%%%%%%%%%%%%%
\end{document}